\documentclass{article}

\pdfoutput=1

\usepackage[utf8]{inputenc} 
\usepackage[T1]{fontenc}    
\usepackage[colorlinks=true, citecolor = blue]{hyperref}       
\usepackage{url}            
\usepackage{booktabs}       
\usepackage{amsfonts}       
\usepackage{nicefrac}       
\usepackage{microtype}      
\usepackage{xcolor}         
\usepackage{geometry}

\usepackage[ruled,vlined,noresetcount]{algorithm2e}

\title{Refined Mechanism Design for Approximately Structured Priors via Active Regression}
\date{}
\author{%
  Christos Boutsikas \\
  Purdue University\\
  \texttt{cboutsik@purdue.edu}\\ 
  \and
  Petros Drineas \\
  Purdue University\\
  \texttt{pdrineas@purdue.edu}\\ 
  \and
  Marios Mertzanidis \\
  Purdue University\\
  \texttt{mmertzan@purdue.edu}\\ 
  \and
  Alexandros Psomas \\
  Purdue University\\
  \texttt{apsomas@cs.purdue.edu}\\ 
  \and
  Paritosh Verma \\
  Purdue University\\
  \texttt{verma136@purdue.edu}\\ 
}

\usepackage{comment}
\usepackage{amsmath,amssymb,amsfonts,mathtools}
\usepackage{amsthm, thmtools}
\usepackage[capitalize]{cleveref}
\usepackage{bbm}
\usepackage{enumitem}
\input{Definitions}

\begin{document}

\maketitle

\begin{abstract}
We consider the problem of a revenue-maximizing seller with a large number of items $m$ for sale to $n$ strategic bidders, whose valuations are drawn independently from high-dimensional, unknown prior distributions. It is well-known that optimal and even approximately-optimal mechanisms for this setting are notoriously difficult to characterize or compute, and, even when they can be found, are often rife with various counter-intuitive properties. In this paper, following a model introduced recently by Cai and Daskalakis~\cite{cai2022recommender}, we consider the case that bidders' prior distributions can be well-approximated by a topic model. We design an active learning component, responsible for interacting with the bidders and outputting low-dimensional approximations of their types, and a mechanism design component, responsible for robustifying mechanisms for the low-dimensional model to work for the approximate types of the former component. On the active learning front, we cast our problem in the framework of Randomized Linear Algebra (RLA) for regression problems, allowing us to import several breakthrough results from that line of research, and adapt them to our setting. On the mechanism design front, we remove many restrictive assumptions of prior work on the type of access needed to the underlying distributions and the associated mechanisms. To the best of our knowledge, our work is the first to formulate connections between mechanism design, and RLA for active learning of regression problems, opening the door for further applications of randomized linear algebra primitives to mechanism design.


\end{abstract}

\section{Introduction}

The design of revenue-optimal auctions is a central problem in Economics and Computer Science. In this problem, a revenue-maximizing seller has $m$ heterogeneous items for sale to $n$ strategic bidders. Each bidder $i$ has a type $\tb_i \in \R^d$ which contains enough information to encode the bidder's willingness to pay for every subset of items. Bidders' types are private information, and thus, in order to provide meaningful guarantees on the seller's revenue, the standard approach in Economics is to make a Bayesian assumption: types are drawn from a joint distribution $\D$. 

Assuming a single item for sale and bidders' types that are drawn independently from known distributions, Myerson's seminal work~\cite{myerson1981optimal} provides a closed-form solution for the revenue-optimal mechanism. Beyond this single-item case, however, multi-item mechanism design remains an active research agenda, even forty years later. Optimal mechanisms are no longer tractable, in any sense of the word, as well as exhibit various counter-intuitive properties~\cite{manelli2007multidimensional,daskalakis2013mechanism,daskalakis2015strong,briest2015pricing,hart2015maximal,hart2019selling,psomas2022infinite}; see~\cite{daskalakis2015multi} for a survey. On the other hand, significant research effort has culminated in numerous compelling positive results, such as simple and approximately optimal auctions~\cite{ChawlaHK07,ChawlaHMS10,ChawlaMS15,Yao15,RubinsteinW15,ChawlaM16,cai2016duality,CaiZ17,kleinberg2019matroid,babaioff2020simple}, as well as efficient algorithms for computing near-optimal auctions~\cite{alaei2012bayesian,cai2012algorithmic,cai2012optimal,cai2013reducing,cai2013understanding}, even with just sampling access to the type distribution~\cite{cole2014sample,huang2015making,devanur2016sample,morgenstern2016learning,cai2017learning,guo2019settling,Brustle2020Robust,gonczarowski2021sample}.


Despite all this progress, however, there are key challenges in applying these results in practice. First, the computational complexity, sample complexity, approximation guarantees, and communication complexity (i.e., the amount of information the bidder should communicate to the seller) of these results often depend on the number of items $m$, which could be prohibitively large (e.g., think of $m$ as the number of items on \texttt{Amazon.com}). Second, bidders' type distributions are typically high-dimensional, or otherwise complex objects, that are not known nor can they be sampled. Instead, the designer might have an estimated distribution $\Dhat$, e.g., through market research, that is close to the real distribution $\D$. Motivated by these issues, Cai and Daskalakis~\cite{cai2022recommender} introduce a model where the true type distribution $\D_i$ of bidder $i$ is close to a structured distribution $\Dhat_i$. Specifically, they assume that there is an underlying design matrix $\Ab \in \R^{m \times k}$ of $k$ ``archetypes,'' with $k \ll m$. Intuitively, bidder $i$ can be approximated by a linear combination of $k$ archetypal bidders. This same assumption has been central in the study of recommender systems. In this model,~\cite{cai2022recommender} give a framework for transforming a mechanism $\Mhat$ for the low-dimensional distribution $\Dhat_z$ into a mechanism for the true type distribution with good revenue guarantees, and whose computational and communication complexity does not depend on the number of items $m$.

The impact of their work notwithstanding, the above results require strong structural assumptions on: the design matrix $\Ab$; the bidders' valuation functions; and very specific access to (or exact knowledge of) the structured distribution $\Dhat_z$ and the mechanism $\Mhat$. \textit{Our work connects the recommender system approaches for mechanism design with recent progress on Randomized Linear Algebra for active learning for regression problems. We relax, and even remove these restrictive assumptions, and open the door to future exploration of more elaborate recommender system models in the context of mechanism design, using randomized linear algebra primitives.}

\noindent\vspace{0.03in}\textbf{The framework and results~\cite{cai2022recommender}.} To place our results in context, we start with a brief overview of~\cite{cai2022recommender}, which considers a setting where the type distribution $\D_i$ of bidder $i$ is close to a distribution $\Dhat_i$ in the Prokhorov distance, under the $\ell_{\infty}$ norm (\Cref{dfn: Prokhorov distance}). Here, $\Dhat_i$ first samples a vector $\zb \in [0,1]^k$ from a low-dimensional distribution $\Dhat_{z,i}$, and then outputs $\Ab \zb$, where $\Ab \in \R^{m \times k}$ is a known matrix. The proposed framework has a learning and a mechanism design component.  

Their learning component consists of a communication-efficient\footnote{By efficient we mean a query protocol that asks each bidder a small amount of queries. Se also Definition~\ref{query_def}.} query protocol $\Qcal$ for interacting with each bidder $i$ such that, if the type $\tb_i$ of bidder $i$ satisfies $\norm{\tb_i - \Ab \zb}_{\infty} \leq \varepsilon$, the query protocol outputs a vector $\Qcal(\tb_i)$ such that $\norm{\Qcal(\tb_i) - \zb}_\infty \leq \zeta_{\infty}$.\footnote{Recall that for any (integer) $1\leq p < \infty$ and a vector $\xb \in \mathbb{R}^d$, $\norm{\xb}_p^p = \sum_{i=1}^d |\xb_i|^p$; for $p=\infty$, 
$\norm{\xb}_\infty = \max_{i=1\ldots d} |\xb_i|$. See~\cite{cai2022recommender} for an exact expression for $\zeta_{\infty}$.}~\cite{cai2022recommender} give such query protocols under strong (distinct) conditions on $\Ab$, and specifically when $\Ab$: \textit{(i)} satisfies an assumption similar to the separability condition of~\cite{donoho2003does}, \textit{(ii)} is generated from a distribution where each archetype is an independent copy of an $m$-dimensional Gaussian, or \textit{(iii)} is generated from a distribution where each archetype is an independent copy of an $m$-dimensional, bounded distribution with weak dependence. We discuss these restrictions in more detail in Appendix~\ref{app: discussion of A}. The query complexity, as well as the error $\zeta_{\infty}$, depend on them, but, importantly, are independent of the number of items $m$.


Their mechanism design component is a refinement of a robustification result of~\cite{Brustle2020Robust}. For this transformation to work, one needs to interact with the mechanism and the underlying distributions using highly non-trivial operations, which are computationally demanding, and require exact knowledge of bidders' valuation functions. In our work we overcome this issues by developing new reductions and plugging them in the framework established by~\cite{Brustle2020Robust} and \cite{cai2022recommender}. The overall interplay between the different mechanism design and active regression components can be seen in \cref{fig:mesh1}.

Combining the two components,~\cite{cai2022recommender} obtain mechanisms for $\D$, for constrained-additive bidders.\footnote{A function is constrained-additive if $v(t,S) = \underset{T \in \mathcal{I} \cap 2^S}{\max} \sum_{j \in T} t_j$, where $\mathcal{I}$ is a downward-closed set system.} In these mechanisms, each bidder is required to answer a small number of simple \texttt{Yes/No} queries of the form ``are you willing to pay $p$ for item $j$?'', such that the loss in revenue and the violation in incentive compatibility do not depend on the number of items $m$.

\subsection{Our contributions}
\noindent \vspace{0.03in}\textbf{Randomized Linear Algebra (RLA) for active learning.} RLA for active learning has focused on solving regression problems of the form 
$$\zb = \arg \min_{\xb}\norm{\tb - \Ab \xb}_{p},$$ 
for $\ell_p$ norms with $1 \leq p < \infty$, by querying \textit{only} a subset of elements of the vector $\tb$. In prior work on RLA for active learning, the focus has been on recovering an approximate solution that achieves a relative error or constant factor approximation to the optimum loss. We adapt these bounds to our setting, which could include noisy instead of exact queries, and prove bounds for the $\ell_p$ norm error of the exact minus the approximate solution. Specifically, we bound
$\norm{\Qcal(\tb_i) - \zb}_p \leq \zeta_p.$
We provide bounds on $\zeta_p$ for all integers $1 \leq p < \infty$. Our bounds depend on the modelling error $\norm{\tb - \Ab \xb}_{p}$ and some measure of the query noise (see Definitions~\ref{dfn: query noise} and~\ref{dfn: model error}); both dependencies are expected. Importantly, our bounds hold for \textit{arbitrary archetype matrices}, very much unlike the work of~\cite{cai2022recommender}, which focused on very restricted classes of matrices. A single property of the archetype matrix, the smallest singular value with respect to the $\ell_p$ induced matrix norm, $\sigma_{\min,p}(\Ab)$, can characterize the quality of the error. As $\sigma_{\min,p}(\Ab)$ decreases, the error $\zeta_p$ grows. Intuitively, $\sigma_{\min,p}(\Ab)$ is a measure of independence of the archetypes, with small values corresponding to linearly dependent archetypes. The query complexity needed to achieve error $\zeta_p$ is almost linear (up to logarithmic factors) on $k$ for $p=1$ and $p=2$, and grows with $k^{\nicefrac{p}{2}}$ for $p\geq 3$. Our query complexity bounds have no dependency on $m$ (number of items) or $d$ (dimensionality of a type) enabling us to produce results way beyond constrained-additive valuations, as $d = 2^m$ dimensions suffice to encode \emph{arbitrary} valuation functions.

It is critical to highlight that our ability to provide bounds on the approximation error for arbitrary archetype matrices is, at least partly, due to leveraging information from the archetype matrix $\Ab$. Specifically, we use this information to select which bidder types to query, instead of just querying types uniformly at random. This information involves the computation or approximation of the well-studied leverage scores of $\Ab$ for $p=2$ and of the so-called Lewis weights for all other values of $p$ (see Section~\ref{sec: proof of machine learning}). We do note that in our framework, the errors due to the modeling of the bidder type $\tb$ as the product $\Ab \zb$ and the query noise are always bounded by the respective $\ell_p$ norm. Thus, our models are more restrictive compared to the $\ell_\infty$ norm models of~\cite{cai2022recommender}. However, to the best of our knowledge, even assuming such restrictive models, the results and tools of prior work do not extend to arbitrary archetype matrices. This is precisely the gap that is bridged by our work, using RLA for active learning of $\ell_p$ norm regression for $1 \leq p <\infty$.

\noindent \vspace{0.03in}\textbf{Mechanism Design.}
On the mechanism design front, our main contribution is relaxing the assumptions of~\cite{cai2022recommender,Brustle2020Robust} on the type of access needed to the low-dimensional distribution and the mechanism for it. Specifically, we further refine the robustification result of Brustle et al.~\cite{Brustle2020Robust} and remove the need for the aforementioned strong oracle. 

The main difficulty of transforming mechanisms for one distribution into mechanisms for another distribution is that the two distributions might not share the same support. The crux of the issue is that the incentive constraints are very delicate; a small change in the underlying distribution may drastically change the agents' \emph{valuation distribution} over the mechanisms' outcomes. One way to tame the distribution of outcomes is to map bids that are not in the support of the initial distribution, to bids that are. Brustle et al.~\cite{Brustle2020Robust} do this by ``optimally misreporting'' on behalf of the bidder, by calculating $\argmax_{\tb'_i \in supp(\Dtilde_{z,i})} \EE_{\bb_{-i} \sim \Dtilde_{z,-i}} [u_i(v_i, \Mhat(\tb'_i,\bb_{-i}))]$, where $\Dtilde_{z,i}$ is a rounded-down version of $\Dhat_{z,i}$. As we've discussed, for this operation to be viable, many things need to be assumed about what the designer knows and can compute about $\Mhat$, $\Dhat_{z,i}$, and the bidder's valuation function. Our approach, instead, leverages the fact that when two distributions are close in Prokhorov distance, under any $\ell_p$ norm, \emph{any} point on the support of one distribution is close to a point on the support of the other, with high probability. Our construction simply maps a report $\wb_i$ to the ``valid'' report (approximately) closest to $\wb_i$ in $\ell_p$ distance. This operation is linear on the support size. Furthermore, our overall robustification result holds for all norms, not just $\ell_{\infty}$, and our construction is completely agnostic to bidders' valuation functions.

\noindent \vspace{0.03in}\textbf{Combining the components.}
Combining the two components we can, without any assumptions on $\Ab$, given a mechanism for the low-dimensional prior, design mechanisms with comparable revenue guarantees, where each bidder is required to answer a small number of queries. Our queries ask a bidder her value for a subset of items, and our mechanism can accommodate \emph{any} valuation function, significantly extending the scope of our results.

\noindent \vspace{0.03in}\textbf{Related Work.} Aside from the work of~\cite{cai2022recommender}, on the mechanism design front, the most relevant works are~\cite{cai2017learning}, that consider learning multi-item auctions given ``approximate'' distributions, and~\cite{Brustle2020Robust}, that consider learning multi-item auctions when type distributions are correlated, yet admit special structure. On the RLA front, we leverage and adapt multiple recent results on approximating $\ell_p$ regression problems in an active learning setting. We discuss prior work on RLA for active learning and its connections to our setting in Section~\ref{sec: proof of machine learning} and Appendix~\ref{sec:appla1}.

\section{Preliminaries}\label{sec: prelims}

Let $[n]\coloneqq \{1, 2, \ldots, n\}$ denote the first $n$ natural numbers. A revenue-maximizing seller has a set $[m]$ of $m$ heterogeneous items for sale to $n$ strategic bidders indexed by $[n]$. Bidder $i$ has a private type vector $\tb_i \in \R^d$, and the types of all bidders are represented by a \emph{valuation profile} $\tb = (\tb_1, \ldots, \tb_n)$. Bidder $i$ has a valuation function $v_i: \R^d \times 2^{[m]} \rightarrow \R_{+}$, that takes as input the bidder's type and a (possibly randomized\footnote{Each subset might be selected according to a probability distribution.}) set of items $S \subseteq [m]$ and outputs the bidder's value for $S$. 
Note that $d \leq 2^m$, since expressing a valuation function requires at most one real number per subset of items.
Types are  drawn independently. Let $\D = \times_{i \in [n]} \D_i$ be the distribution over bidders' types, $\D_{-i} = \times_{j \in [n]/\{i\}} \D_{j}$ be the distribution of all bidders excluding $i$, and $supp(\D)$ be the support of a distribution $\D$.

{\bf Mechanisms.} A mechanism $\M = (x, p)$ is a tuple where $x: \R_{+}^{nd} \rightarrow 2^{[nm]}$ is the allocation rule, and  $p: \R_{+}^{nd} \rightarrow \R_{+}^{n}$ is the payment rule, which map \emph{reported} types to allocations of the items and payments, respectively.  Specifically, $x_{i,j}(\bb) \coloneqq (x(\bb))_{i,j}$ denotes the probability that bidder $i$ receives item $j$ for input valuation profile $\bb$, and $p_i(\bb) \coloneqq (p(\bb))_i$ denotes the amount bidder $i$ has to pay. Let $u_i(\tb_i, \M(\bb))$ be the utility of bidder $i$ with type $\tb_i$ for participating in mechanism $\M$, under reports $\bb$. Bidders are risk-free and quasi-linear i.e., $u_i(\tb_i, \M(\bb)) = \EE \left[v_i(\tb_i, x(\bb))-p_i(\bb) \right]$, where the expectation is taken over the randomness of the allocation rule.  Since we only consider truthful mechanisms, unless stated otherwise, reported types will be the same as the true types.

A bidder's objective is to maximize her utility. The seller strives to design mechanisms that incentivize bidders to report truthfully. We use the following notions of truthfulness. A mechanism $\M$ is \emph{ $\varepsilon$-Bayesian Incentive Compatible ($\varepsilon$-BIC)}, if for each bidder $i \in [n]$, any type $\tb_i$ and misreport $\tb'_i$ we have that: $\EE_{\tb_{-i} \sim \D_{-i}}[u_i(\tb_i, \M(\tb_i, \tb_{-i}))] \ge \EE_{\tb_{-i} \sim \D_{-i}}[u_i(\tb_i, \M(\tb'_i, \tb_{-i}))] - \varepsilon$. A mechanism $\M$ is \emph{$(\varepsilon, \delta)$-BIC} if for each bidder $i \in [n]$, and any misreport $\tb'_i$ we have that:
\[\PP_{\tb_i \sim \D_i} \left[\EE_{\tb_{-i} \sim \D_{-i}}[u_i(\tb_i, \M(\tb_i, \tb_{-i}))] \ge \EE_{\tb_{-i} \sim \D_{-i}}[u_i(\tb_i, \M(\tb'_i, \tb_{-i}))] - \varepsilon \right] \ge 1 - \delta. \]
A $(\varepsilon, 0)$-BIC mechanism is a $\varepsilon$-BIC mechanism; a $0$-BIC mechanism is simply BIC. Finally, a mechanism $\M$ is \emph{ ex-post Individually Rational (IR)} if for all valuation profiles $\tb$ and all bidders $i \in n$, $u_i(\tb_i, \M(\tb_i, \tb_{-i})) \ge 0$. The seller's objective is to maximize her expected \emph{revenue}. For a mechanism $\M$ and distribution $\D$ we denote the expected revenue as $Rev(\M,\D) = \EE_{\tb \sim \D}\left[\sum_{i \in [n]} p_i(\tb) \right]$. Note that, we are calculating revenue assuming truthful reports, even for, e.g., $(\varepsilon, \delta)$-BIC mechanisms.

{\bf Statistical Distance.} In this work we design mechanisms that work well, as long as they are evaluated on distributions that are ``close'' to $\D$. Here, we define the notion of distance between two probability measures that we use throughout the paper.

\begin{definition}[Prokhorov Distance]\label{dfn: Prokhorov distance}
Let $(\Ucal, d)$ be a metric space and $\Bcal$ be a $\sigma$-algebra on $\Ucal$. For $A\in \Bcal$, let $A^{\varepsilon} = \{x: \exists y \in A \text{ s.t. } \pi(x,y)<\varepsilon\}$ where $\pi$ is some distance metric. Two probability measures $P$, $Q$ on $\Bcal$ have Prokhorov distance: $inf\{\varepsilon>0: P(A) \le Q(A^{\varepsilon})+\varepsilon \text{ and } Q(A) \le P(A^{\varepsilon})+\varepsilon, \forall A \in \Bcal\}$. We choose $\pi$ to be the $\ell_p$-distance, and we denote the Prokhorov distance between measures $P, Q$ as $\pi_p(P,Q)$.
\end{definition}

The following is an equivalent characterization of Prokhorov distance due to Strassen \cite{strassen1965}.

\begin{lemma}[\cite{strassen1965}] \label{prokhorovCharacterization}
Let $\D$ and $\D'$ be two distributions supported on $\R^n$. $\pi_p(\D, \D') \le \varepsilon$ iff there exists coupling $\gamma$ of $\D$ and $\D'$, such that $\PP_{(\xb,\yb) \sim \gamma}[\norm{\xb-\yb}_p > \varepsilon] \le \varepsilon$.
\end{lemma}

{\bf Recommendation system-inspired model.} We assume that, for each bidder, there exists a known design matrix $\Ab \in \R^{d \times k}$ \footnote{Notice that we consider the more general case of having $d$ number of rows (instead of $m$).}, where the columns of $\Ab$ represent $k$ ``archetypes'', for a constant $k$. Our results hold if these matrices are different for each bidder, however, for ease of notation we will assume all bidders have the same design matrix. For each bidder $i$ there exists a distribution $\Dhat_{z,i}$ supported on the latent space $[0,1]^k$. Let $\Dhat_i = \Ab \circ \Dhat_{z,i}$ be the distribution induced by multiplying a sample from $\Dhat_{z,i}$ with the design matrix $\Ab$, i.e., $\Dhat_i$ is the distribution of $\Ab\yb$ where $\yb \sim \Dhat_{z,i}$. 

The valuation function over the latent types is defined as $v^\Ab_i(\zb_i, S) \coloneqq v_i(\Ab \zb_i, S)$ for any bundle $S \subseteq [m]$. We will use the following notion of Lipschitz continuity for valuation functions.

\begin{definition}[Lipschitz Valuation]
A valuation function $v(\cdot,\cdot): \R^d \times [0,1]^m \rightarrow \R_+$ is $\Lcal$-Lipschitz, if for any two types $\tb, \tb' \in \R^d$ and any bundle $S \subseteq [m]$, $|v(\tb,S)-v(\tb',S)| \le \Lcal \norm{\tb-\tb'}_{\infty}$. 
\end{definition}

We include a table,~\cref{table: notation}, with all the notation used throughout the paper in the appendix.

\section{Active Learning for Regression and Mechanism Design}\label{section:mainRobustnessResult}

In this section, we state our main results, deferring all technical proofs to the appendix. We present mechanisms that are completely agnostic with respect to $\D$, the distribution from which bidders' types are drawn from. However, we have limited access (described later in this section) to \textit{(i)} a design matrix $\Ab$; \textit{(ii)} distributions $\Dhat_{z,i}$ over $\R^k$, where for all $i \in [n]$, $\pi_p(\D_i, \Ab \circ \Dhat_{z,i}) \le \emdl$ for some $\emdl > 0$; and \textit{(iii)} a mechanism $\Mhat$ for $\Dhat_z = \times_{i\in [n]} \Dhat_{z,i}$.
This limited access to the design matrix motivates the use of active learning, which deals precisely with settings where the algorithm is allowed to (interactively) query a subset of the available data points for their respective labels (see~\cite{chen2021query, musco2022active} for precise definitions of the active learning setting in regression problems). Our approach is modular and starts by building an active learning component for regression problems (Section~\ref{sec: proof of machine learning}) followed by the mechanism design component (Section~\ref{subsec: MD component}). We combine the two components to get an overall mechanism for $\D$ in Section~\ref{subsec:combined}. 


\subsection{Active learning for regression via Randomized Linear Algebra}\label{sec: proof of machine learning}

Our objective is to design a communication-efficient, active learning query protocol for the seller that interacts with each bidder $i$, and infers their type $\tb_i \in \R^d$ by accessing only a small subset of elements of the type vector (as $d$ is very large). We use $\Qcal$ to denote the query protocol, whose output is a vector in the low-dimensional latent space $\R^k$. A bidder interacts with the query protocol \textit{truthfully} if it is in her best interest to evaluate functions requested by the protocol on her true private type $\tb_i$. We use $\Qcal(\tb_i) \in \R^k$ to denote the output of $\Qcal$ when interacting with a truthful bidder with type $\tb_i$. We now define the notion of an $(\emdl, \zeta_p, p)$-query protocol and the notion of \textit{query noise}.

\begin{definition}[$(\emdl, \zeta_p, p)$-query protocol]  \label{query_def}
 $\Qcal$ is called an $(\emdl, \zeta_p, p)$-query protocol, if, for all $\tb \in \R^d$ and $\zb \in \R^k$ satisfying $\norm{\tb-\Ab \zb}_p \le \emdl$, we have $\norm{\zb-\Qcal(\tb)}_p \le \zeta_p$.
\end{definition}
\begin{definition}[Query noise]\label{dfn: query noise}
Let $\tb_i$ be the true type of a bidder. Our query protocol can access entries of $\tb_i + \epsb_{\texttt{nq,p}}$, where $\epsb_{\texttt{nq,p}}$ is an (unknown) vector. The query noise $\enq$ satisfies $\|\epsb_{\texttt{nq,p}}\|_p \leq \enq$.
\end{definition}

%
The query noise depends on the specifics of the interactions between the seller and the bidder. For example, if the seller is only allowed to ask queries of the form ```what is your value for the subset $S$?'', the query noise $\enq$ is equal to zero. Our bounds will also depend on the \textit{model error}.
\begin{definition}[Model error]\label{dfn: model error}
Given a valuation profile $\tb \in \R^{nd}$, the \emph{model error} is $\emdl$ if, for all $i \in [n]$, there exists a $\zb_i \in \R^k$ such that $\norm{\tb_i-\Ab \zb_i}_p \le \emdl$.
\end{definition}

Note that, we don't have bounds of the form ``$\norm{\tb_i-\Ab \zb_i} \leq \varepsilon$'' for individual types, but for the distributions $\D_i$ and $\Ab \circ \Dhat_{z,i}$. The characterization of Prokhorov distance (\Cref{prokhorovCharacterization}) allows us to relate the two quantities in the proofs that follow. 


%
%
%

We now rephrase the above discussion in order to cast it in the framework of Randomized Linear Algebra \textit{(RLA)} and active learning. Dropping the index $i$ for notational simplicity, we assume that $\Ab \zb \approx \tb$ and we seek to recover an approximate solution vector $\Qcal(\tb)$ such that the $\ell_p$ norm error between the approximate and the optimal solution is bounded. \textit{Importantly}, the query protocol $\Qcal$ is \textit{not} allowed full access to the vector $\tb$ in order to construct the approximate solution vector. This is a well-studied problem in the context of RLA: the learner is given a large collection of $k$-dimensional data points (the $d \gg k$ rows of the design matrix $\Ab \in \mathbb{R}^{d \times k}$), but can only query a small subset of the real-valued labels associated with each data point (elements of the vector $\tb \in \mathbb{R}^d$). Prior work in RLA and active learning has studied this problem in order to identify the optimal number of queries that allow efficient, typically relative error, approximations of the loss function. In our parlance, prior work has explored the existence of query protocols that construct a vector $\Qcal(\tb)$ such that 
\begin{equation} \label{rla_reg_0}
    \norm{\tb-\Ab \Qcal(\tb)}_p \leq \gamma_p \norm{\tb-\Ab\zb}_p,
\end{equation}

where $\gamma_p > 1$ is an error parameter that controls the approximation accuracy. Of particular interest in the RLA literature are \textit{relative error} approximations, with $\gamma_p=1+\epsilon$, for some small $\epsilon>0$; see~\cite{musco2021active_arxiv,musco2022active} for a detailed discussion. However, relative error approximations are less important in our setting, since our protocols in Section~\ref{subsec: MD component} necessitate $\zeta_p \geq \emdl$.  
For $p=2$, the underlying problem is active learning for least-squares regression:~\cite{drineas2011faster} analyzed its complexity (namely, the number of queries) of query protocols in this setting, eventually providing matching upper and lower bounds. Similarly, for $p=1$, the underlying problem is active learning for least absolute deviation regression, a robust version of least-squares regression:~\cite{chen2021query} analyzed the complexity of query protocols in this setting.
%
%
The query protocols of~\cite{drineas2011faster, chen2021query} are straightforward: they sample a small set of labels (i.e., bidder types) and elicit the bidder's preferences for this set. Then, the respective $\ell_p$ norm regression problem is solved on the smaller set and the resulting solution is returned as an approximation to the original solution.\footnote{To be precise, multiple smaller problems have to be solved and a ``good enough'' solution has to be chosen in order to boost the success probability. See Appendix~\ref{sec:appla1} for details.}
%
%
The types to be sampled (see Appendix~\ref{appendix:linear-algebra:query12} for details) are selected using distributions that can be constructed by accessing \textit{only} the design matrix $\Ab$. Specifically, for the $p=2$ case, one needs to compute or approximate the \textit{leverage scores} of the rows of the matrix $\Ab$. For the $p=1$ case, one needs to compute or approximate the \textit{Lewis weights} of the design matrix $\Ab$. (The Lewis weights are an extension of the leverage scores to $\ell_p$ norms for $p\neq 2$.) The work of~\cite{derezinski2017unbiased,derezinski2018tail,derezinski2021determinantal,JMLR:v19:17-781,chen2019active} for the $p=2$ case involves more elaborate query protocols, using primitives such as volume sampling and the Batson-Spielman-Srivastava sparsifier to improve the query complexity. Finally, the $p>2$ case for active learning for regression problems was recently resolved in~\cite{musco2022active, musco2021active_arxiv}; we discuss their approach in our context in Appendix~\ref{appendix:linear-algebra:querylargerthan2}.

To the best of our knowledge, our work is the first one to formulate connections between mechanism design and Randomized Linear Algebra for active learning. Two technical points of departure that are needed in order to adapt the RLA work for active learning to the mechanism design framework are: \textit{(i)} we need to derive bounds of the form of eqn.~(\ref{rla_reg_0}) for the $\ell_p$ norm distance between the exact and approximate solutions, whereas prior work typically bounds the error of the \textit{loss} function when an approximate solution is used; and \textit{(ii)} the entries of the bidder's type vector $\tb$ might not be known exactly, but only up to a small error. The latter assumption corresponds to the use of \textit{noisy queries} in the model of~\cite{cai2022recommender} and is known to be equivalent, up to logarithmic factors, to \textit{threshold queries} via binary search. Our work addresses both technicalities and seamlessly combines the RLA work for active learning with mechanism design.

Prior to stating our main result, 
%
%
%
we need to define a fundamental property of the design matrix $\Ab \in \mathbb{R}^{d \times k}$ that will affect the approximation error. Let  
\begin{equation} \label{gen_sigma}
    \sigma_{\min,p}(\Ab) =  \min_{\xb \in \R^{k},\ \|\xb\|_p=1} \norm{\Ab \xb}_{p}.
\end{equation}
For $p=2$, this is simply the smallest singular value of the matrix $\Ab$. For other values of $p$, the above definition is the standard generalization of the smallest singular value of $\Ab$ for the induced matrix $\ell_p$ norm. Notice that $\sigma_{\min,p}(\Ab)$ is a property of the matrix $\Ab$ and can be computed \textit{a priori} via, say, the QR factorization or the Singular Value Decomposition (SVD) for $p=2$ and via linear programming for $p=1$. As we will see in Theorem~\ref{linear-algebra-thm} below, smaller values of $\sigma_{\min,p}(\Ab)$ result in increased sample complexity for our query protocols.

\begin{theorem} \label{linear-algebra-thm}
Let $\Ab \in \mathbb{R}^{d \times k}$ be the design matrix, and recall the definitions of the model error $\varepsilon_{\texttt{mdl},p}$ (Definition~\ref{dfn: model error}) and the query noise (Definition~\ref{dfn: query noise}). For all integers $1\leq p < \infty$, there exist query protocols $\Qcal$ using $s_p$ queries for each bidder $i \in [n]$, such that, with probability at least $1-\delta$,
$$\norm{\zb_i-\Qcal(\tb_i)}_p \le \dfrac{c_p (\varepsilon_{\texttt{mdl},p}+\varepsilon_{\texttt{nq,p}})}{\sigma_{\min,p}(\Ab)} = \zeta_p$$
holds for all $n$ bidders $i\in [n]$. Here $c_p$ is a small constant that depends on $p$.\footnote{\label{fnref1}We make no attempt to optimize constants and focus on simplicity of presentation. In our proofs, $c_1 = 2.5$; $c_2=7.5$; and for $p\geq 3$, $c_p = 18\cdot(200)^{\nicefrac{1}{p}}+3$. Notice that the last constant converges to 21 as $p$ increases.} The respective query complexities for $p=1$ and $p=2$ are (asymptotically) identical: 
\begin{flalign} 
%
s_1 = s_2 = O\left(k \cdot \ln k \cdot \ln \nicefrac{n}{\delta}\right).\label{eqn:complexitypequalonetwo}
%
%
\end{flalign}
For $p \geq 3$, the query complexity is
\begin{flalign} 
s_p = O\left(k^{\nicefrac{p}{2}}\cdot \ln^3 k \cdot \ln \nicefrac{n}{\delta}\right).\label{eqn:complexitygeneralp}
\end{flalign}

%
%
\end{theorem}
Several comments are in order. 
\textit{(i)} The error $\zeta_p$ is a small constant times the modelling error plus the error due to noisy queries. In the limit case where the modelling error is equal to zero and the queries are noiseless, the bidders' types can be recovered exactly in our framework. However, as the modelling error and the query noise increase, approximating user types becomes harder and less accurate.
\textit{(ii)} Importantly, the approximation accuracy of Theorem~\ref{linear-algebra-thm} grows linearly with the inverse of the smallest $\ell_p$ norm singular value of the design matrix $\Ab$. Our results indicate that the approximation accuracy of the query model $\Qcal$ depends on this simple property of the design matrix $\Ab$. For example, focusing on the $p=2$ case, our theorem shows that as the archetypes (columns of the matrix $\Ab$) become more linearly dependent and the smallest singular value approaches zero, the error of our approximation worsens. This is quite reasonable: if archetypes are linearly dependent, then it is increasingly difficult to approximate the respective entries of the vector $\zb$. 
\textit{(iii)} The query complexities $s_1$ and $s_2$ are asymptotically identical, growing linearly with $k \ln k$, where $k$ is the number of archetypes. They both depend on the log of the number of bidders (due to a union bound) and on the log of $\nicefrac{1}{\delta}$, where $\delta$ is the failure probability. The query complexity for $p \geq 3$ is larger and is dominated by the $k^{\nicefrac{p}{2}}$ term. Importantly, the query complexity remains independent of $d$, the number of bidder types, which, in worst case, could be exponential to the number of underlying items.
%
%
%
\textit{(iv)} Improving the sampling complexities $s_1$ and $s_2$ has been a topic of intense interest in the RLA community and we defer the reader to~\cite{musco2021active_arxiv,musco2022active}, which has essentially provided matching upper and lower bounds for various values of $p$. We just note that for the well-studied $p=2$ case, volume sampling approaches~\cite{derezinski2017unbiased,derezinski2018tail,derezinski2021determinantal,JMLR:v19:17-781} achieve essentially matching bounds, while the work of~\cite{chen2019active} removes (at least in expectation) the $\ln k$ factor from $s_2$, at the expense of significant additional protocol complexity. From a technical perspective, we note that $\zeta_p \geq \emdl$, as necessitated in~\cref{mainTheorem} and that our query protocols are all one-round protocols.

Finally, notice that our theorem works for all $p \geq 1$, but not for $p=\infty$, which was the setting of~\cite{cai2022recommender}. In Appendix~\ref{appendix:cdlemma}, we present a (modest) improvement of the result of~\cite{cai2022recommender} and explain why it seems improbable that the $p=\infty$ case can be generalized to a much broader class of design matrices. This is a strong motivating factor to explore properties of mechanism design for the recommender system setting for other $\ell_p$ norms, as we do in this work.

\subsection{The Mechanism Design component}\label{subsec: MD component}

The goal of the mechanism design component is to transform a mechanism $\Mhat$ for $\Dhat$ into a mechanism $\M$ for $\D_z$. We first define exactly the type of access to $\Dhat_z$ and $\Mhat$ our construction requires.

\begin{definition}[Access to $\Mhat$]\label{dfn: access to Mhat}
    By ``query access to $\Mhat$'' we mean access to an oracle which, given a valuation profile $\tb$, outputs the allocation and payments of $\Mhat$ on input $\tb$. 
\end{definition}

\begin{definition}[Access to $\Dhat_z$]\label{dfn: sampling dhat}
By ``oracle access to $\Dhat$'' we mean access to (1) a sampling algorithm $\Scal_i$ for each $i \in [n]$, where $\Scal_i(\xb, \delta)$ draws a sample from the conditional distribution of $\Dhat_{z,i}$ on the $k$-dimensional cube $\times_{j \in [k]}[x_j, x_j+\delta_j)$, and (2) an oracle which, given as input a type $\tb_i$ for bidder $i$, outputs the type in the support of $\Dhat_{z,i}$ that is closest to $\tb_i$ in $\ell_p$ distance, i.e., outputs $argmin_{\tb'_i \in supp(\Dhat_{z,i})} \norm{\tb_i - \tb'_i}_p$. 
\end{definition}

If the allocation is randomized, our approach works even if the query to the oracle returns a (correct) deterministic instantiation of the randomized allocation.\footnote{That is, if, for example, $\Mhat$ allocates item $j$ to bidder $i$ with probability $1/2$ and with the remaining probability item $j$ is not allocated, our construction does not need to know this distribution/fractional allocation and works even if nature samples and returns an integral allocation for item $j$.}

In~\cref{dfn: sampling dhat}, the first part of our oracle access (sampling from the conditional) is also necessary in~\cite{cai2022recommender}. The second part is new to our work, and replaces a strong requirement in~\cite{cai2022recommender}. In more detail, given a type $\tb_i$, Cai and Daskalakis~\cite{cai2022recommender} (as well as Brustle et al.~\cite{Brustle2020Robust}) need to know if $\tb_i \in supp(\Dtilde_{z,i})$, and, if not, need access to $argmax_{\tb'_i \in supp(\Dtilde_{z,i})} \EE_{\bb_{-i} \sim \Dtilde_{z,-i}} [u_i(v_i, \Mhat(\tb'_i,\bb_{-i}))]$ where $\Dtilde_{z,i}$ is a rounded-down version of $\Dhat_{z,i}$. However, for arbitrary distributions and mechanisms, this task might be computationally inefficient, or simply infeasible. In our work, we need access to $argmin_{\tb'_i \in supp(\Dhat_{z,i})} \norm{\tb_i - \tb'_i}_p$ in the ``no'' case. 

Given these definitions, our main theorem for this component is stated as follows.

\begin{theorem}\label{mainTheorem}
Let $\D = \times_{i=1}^n \D_i$ be the bidders' type distribution and $v_i: \R^d \times 2^{[m]} \rightarrow \R_+$ be a $\Lcal$-Lipschitz valuation function for each bidder $i \in [n]$. Also, let $\Ab \in \RR{d}{k}$ be a design matrix and $\Dhat_z = \times_{i=1}^n \Dhat_{z,i}$, where $\Dhat_{z,i}$ is a distribution over $\R^k$ for each $i \in [n]$.

Suppose that we are given (1) query access to a mechanism $\Mhat$ that is IR and BIC w.r.t. $\Dhat_z$ and valuations $\{v_i^\Ab\}_{i \in [n]}$, (2) oracle access to $\Dhat_z$, and (3) any $(\emdl, \zeta_p, p)$-query protocol $\Qcal$ with $\zeta_p \ge \emdl$. Then, we can construct a mechanism $\M$ that is oblivious to $\D$ and $v(\cdot,\cdot)$, such that for all $\D$ that satisfy $\pi_p(\D_i, \Ab \circ \Dhat_{z,i}) \le \emdl$ for all $i \in [n]$, the following hold: (1) $\M$ only interacts with every bidder using $\Qcal$ once, (2) $\M$ is IR and $(\eta, \mu)$-BIC w.r.t. $\D$, where $\mu = O(\sqrt{\zeta_p})$ and $\eta = O(n\norm{\Ab}_{\infty} \Lcal \sqrt{\zeta_p})$, and (3) the expected revenue of $\M$ is at least $Rev(\Mhat, \Dhat_z) - O(n \eta)$.
\end{theorem}
Note that $\M$ is an indirect mechanism, so it is slightly imprecise to call it $(\eta, \mu)$-BIC. Formally, interacting with $\Qcal$ truthfully is an approximate Bayes-Nash equilibrium. 

In order to prove~\Cref{mainTheorem}, we use a key lemma,~\cref{mainMechanism}, which establishes the robustness guarantees of~\Cref{mainTheorem}, but in the space of latent types. Intuitively, let $\tb_i$ be the type of bidder $i$, and $\zb_i$ be a random variable distributed according to $\Dhat_{z,i}$. We know that $\pi_p(\Dcal_i, \Ab \circ \Dhat_i) \le \emdl \le \zeta_p$. Due to~\cref{prokhorovCharacterization}, there exists a coupling such that with probability greater than $1-\zeta_p$, $\norm{\tb_i - \Ab \zb_i}_{p} \le \zeta_p$. Since the seller uses a $(\emdl, \zeta_p,p)$-query protocol, with probability at least $1-\zeta_p$, $\norm{\Qcal(\tb_i) - \zb_i}_{p} \le \zeta_p$. Note that, this implies that $\zb_i$ and $\Qcal(t_i)$ are distributed such that their Prokhorov distance is at most $\zeta_p$. At this step,~\cref{mainMechanism} provides us a mechanism $\Mtilde$, constructed from $\Mhat$, that we can execute on types $\Qcal(\tb_1), \ldots, \Qcal(\tb_n)$, obtained by interacting with the bidders via the query protocol. With probability at least $1-\zeta_p$, we have that $\norm{\tb_i - \Ab \zb_i}_{\infty} \le \norm{\tb_i - \Ab \zb_i}_{p} \le \emdl$ and thus, using the fact that the query protocol ensures $\norm{\Qcal(\tb_i) - \zb_i}_{p} \le \zeta_p$ as well, we have $\norm{\tb_i - \Ab \Qcal(\tb_i)}_{\infty} \le \emdl + k\norm{\Ab}_{\infty}\zeta_p$. The guarantees of $\Mtilde$ for the distribution over $\Qcal(\tb_i)$s are therefore translated into guarantees (with a small error) of the overall mechanism for the $\D$.


The proof of~\cref{mainMechanism} is quite involved and is the main focus of our analysis. Here, we sketch the key ideas behind the proof, and defer all formal arguments to Appendix~\ref{app:mechanism design}. The proof uses the following notion of a rounded distribution.
\begin{definition}[Rounded Distribution]\label{dfn: rounded}
    Let $\Fcal$ be a distribution supported on $\R_{\ge 0}^k$. For any $\delta > 0$ and $\ell\in [0, \delta]^k$, we define  function $r^{(\ell,\delta)}: \R_{\ge 0}^k \mapsto \R_{\ge 0}^k$ such that $r^{(\ell,\delta)}_i(\xb) = max\left\{ \floor{\frac{x_i - \ell_i}{\delta}}\cdot \delta + \ell_i, 0\right\}$ for all $i \in [k]$. Let $\xb$ be a random vector sampled from $\Fcal$. We define $\floor{\Fcal}_{\ell,\delta}$ as the distribution for the random variable $r^{(\ell,\delta)}(\xb)$, and we call $\floor{\Fcal}_{\ell,\delta}$ the rounded distribution of $\Fcal$.
\end{definition}
We follow an approach similar to Brustle et al.~\cite{Brustle2020Robust}. The main idea is that arguing directly about mechanisms for distributions that are close in Prokhorov distance is difficult. On the flip side, arguing about mechanisms for distributions that are close in total variation distance is much easier, since the total variation is a more stringent (and hence more tamable) notion of distance. The key observation is that, if two distributions are close in Prokhorov distance then, in expectation over the random parameter $\ell$, their rounded-down versions are close in total variation distance.  

Our overall construction is via three reductions. First, in~\Cref{roundingDown}, given a mechanism for $\Fhat_z$ we design a mechanism for the rounded-down version. Second, in~\Cref{TVrobustness}, given a mechanism for the rounded-down $\Fhat_z$ we design a mechanism for $\floor{\Fcal_z}_{\ell,\delta}$, which maintains its guarantees if $\pi_p(\Fcal_z,\Fhat_z)$ is small. Third, in~\Cref{RoundingUp}, given a mechanism for $\floor{\Fcal_z}_{\ell,\delta}$ we design a mechanism for $\Fcal_z$. \cref{fig:mesh1} presents a detailed overview  of the overall design architecture, and how the RLA and different mechanism design components interact.

\begin{figure}[ht!]
    \centering
    \includegraphics[scale=0.75]{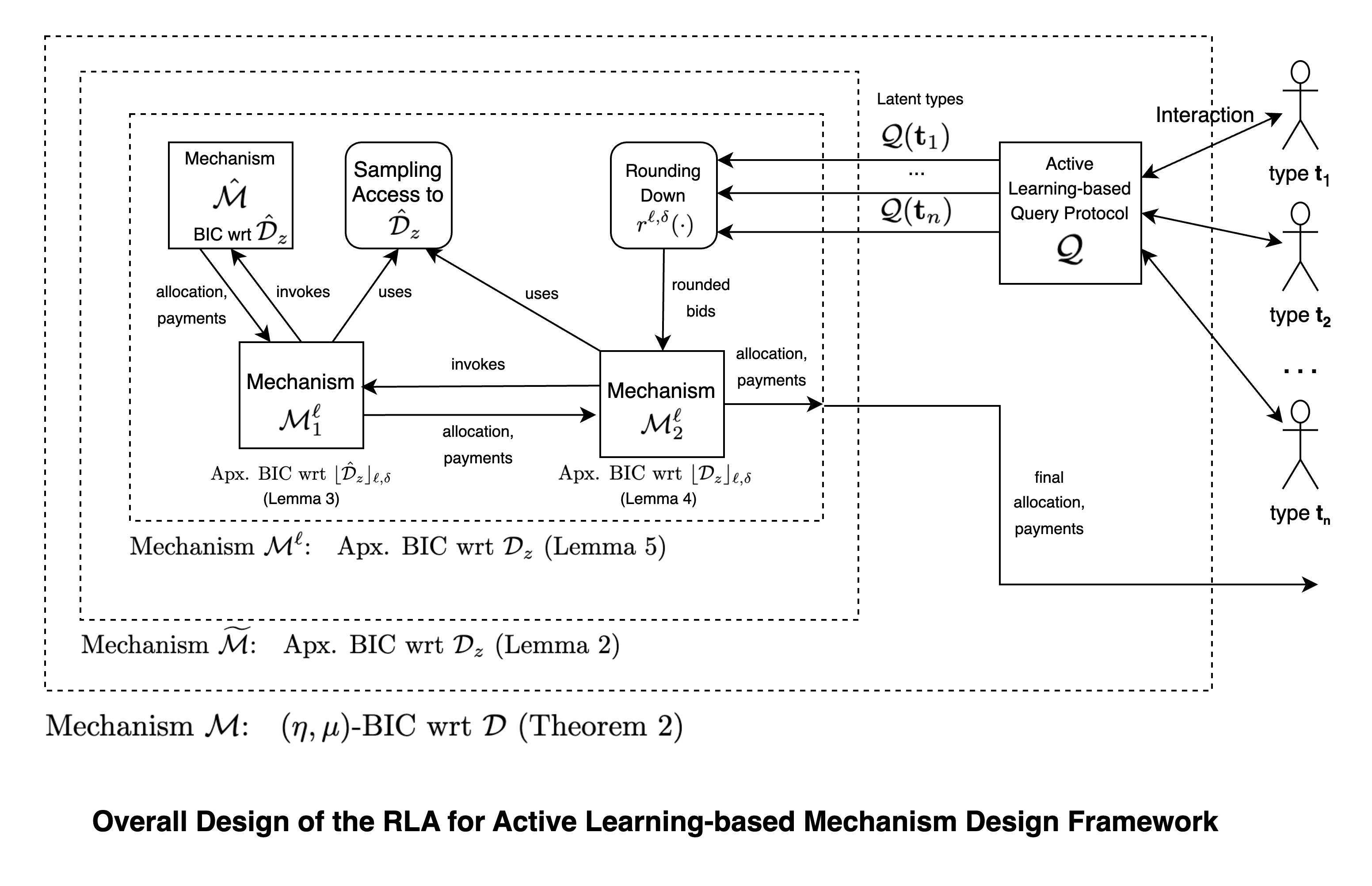}
    \caption{Agents interact with the query protocol $\mathcal{Q}$, which learns their latent types $\mathcal{Q}(\mathbf{t}_i)$s. The mechanism design component (which is oblivious to the distributions of the true agents' types $\mathcal{D}$) then uses these to produce the final allocation and payments, utilizing only query access to $\hat{\mathcal{M}}$ and sampling access to $\hat{\mathcal{D}}_z$ such that the overall framework is approximately $(\eta, \mu)$-BIC wrt $\mathcal{D}$.}
    \label{fig:mesh1}
\end{figure}

Our proofs for Lemmas~\ref{roundingDown} and~\ref{RoundingUp} are adaptations of the corresponding lemmas of~\cite{Brustle2020Robust}, where our main task is to flesh out the exact dependence on the dimensionality of the latent space and the $\ell_p$-norm (versus the $\ell_{1}$-norm in~\cite{Brustle2020Robust}). The novelty of our approach comes in the construction and analysis of Lemma~\ref{TVrobustness}. The difficulty of transforming mechanisms for $\floor{\smash{\Fhat_{z}}}_{\ell, \delta}$ into mechanisms for $\floor{\smash{\Fcal_{z}}}_{\ell, \delta}$ is that the two distributions might not share the same support. Thus, we need a way to map bids that are not in the support of $\floor{\smash{\Fhat_z}}_{\ell,\delta}$ to bids that are. Brustle et al.~\cite{Brustle2020Robust} do this by ``optimally misreporting'' on behalf of the bidder, by calculating $argmax_{\zb \in supp(\floor{\Fhat_{z,i}}_{\ell, \delta}) } \EE_{b_{-i} \sim \floor{(\Fhat_z)_{-i}}_{\ell, \delta}} [u_i(v_i, \Mhat(\zb,\bb_{-i}))]$, and then picking matching payments that make the overall mechanism IR. Our approach leverages the fact that $\Fhat_{z,i}$ and $\Fcal_{z,i}$ are close in Prokhorov distance, and thus any point on the support of one distribution is close to a point on the support of the other, with high probability. An ideal construction would map a report $\wb_i$ to the ``valid'' report (i.e., a report in the support of $\floor{\smash{\Fhat_{z,i}}}_{\ell, \delta}$) 
that minimizes the $\ell_p$ distance to $\wb_i$. This operation is linear on the support of $\floor{\smash{\Fhat_{z,i}}}_{\ell, \delta}$, and does not need any information on the valuation functions, nor on the actual probabilities that elements of the distribution are sampled with. Unfortunately, our assumption on what ``oracle access'' means does not allow us to do this operation (finding the closest point w.r.t. $\ell_p$) on $\floor{\smash{\Fhat_{z,i}}}_{\ell, \delta}$, but only on $\Fhat_{z,i}$; we prove that, by occurring a small loss, our assumption suffices. 

\subsection{Putting everything together}\label{subsec:combined}
Combining Theorems~\ref{linear-algebra-thm} and~\ref{mainTheorem}, we can give mechanisms for concrete settings. 
Formally, we have the following theorem.

\begin{theorem}\label{thm: combined results}
Under the same setting as in~\Cref{mainTheorem}, for bidders with arbitrary valuation functions, we can construct mechanism $\M$ using only query access to the mechanism $\Mhat$ (\cref{dfn: access to Mhat}) and oracle access to distribution $\Dhat$ (\cref{dfn: sampling dhat}), and oblivious to the true type distribution $\D$. We consider queries (to each bidder $i$) of the form ``What is your value for the subset of items $S$?'' 

Mechanism $\M$ is IR and $(\eta, \mu)$-BIC w.r.t. $\D$, where $\mu = O(\sqrt{\zeta_p})$ and $\eta = O(n\norm{\Ab}_{\infty} \Lcal \sqrt{\zeta_p})$, and the expected revenue of $\M$ is at least $Rev(\Mhat, \Dhat_z) - O(n\eta)$. Additionally, with probability at least $1-\delta$, 
$$\zeta_p = c_p \cdot \left( \sigma_{\min,p}(\Ab) \right)^{-1} \cdot \varepsilon_{\texttt{mdl},p}$$ 
for a small constant $c_p$ that depends on the parameter $p$ (see footnote~\ref{fnref1}).
%
%
The number of queries is $O\left(k\cdot \ln k\cdot \ln (\nicefrac{n}{\delta})\right)$ and $O\left(k^{\nicefrac{p}{2}}\cdot\ln^3 k\cdot \ln(\nicefrac{n}{\delta})\right)$ for $p=1,2$ and $p \geq 3$, respectively.

\end{theorem}

The proof of~\Cref{thm: combined results} follows from~\Cref{mainTheorem} and~\Cref{linear-algebra-thm}, and is deferred to Appendix~\ref{app: combined proof}.

As we've already discussed, the main mechanism of Cai and Daskalakis~\cite{cai2022recommender} requires bidders to have constrained-additive valuations\footnote{A valuation function is constrained-additive if $v(t,S) = max_{T \in \mathcal{I} \cap 2^S} \sum_{j \in T} t_j$, where $\mathcal{I}$ is a downward-closed set system.}, as well $\Ab$ to satisfy a number of restrictions. Here, we completely remove both conditions. On the flip-side,~\cite{cai2022recommender} ask bidders weaker queries, of the form ``are you willing to pay price $\tau$ for item $j$?'' Using such queries, one can binary search over $\tau$, and drive down the query noise (see~\Cref{dfn: query noise}). For $\ell_{\infty}$, the extra cost of such an operation would be $\ln\left( \norm{\Ab}_{\infty}/\varepsilon \right)$, where $\varepsilon$ is the desired accuracy. However, for other $p$-norms, for the same target accuracy, this operation requires an extra factor of $\Theta( \log( d^{1/p} ) )$ queries, giving a dependence on $d$.


\section{Conclusions and Future Work}

In this paper, we study mechanism design for prior distributions close to a topic model, inspired from the recommender systems literature. We formulate connections between mechanism design and Randomized Linear Algebra for active learning for regression problems, importing state-of-the-art results from Randomized Linear Algebra to mechanism design, and alleviate or relax restrictive assumptions of prior work. Developing a deeper understanding of such connections is an important direction for future research. For example, one could study this and other topic models in the context of mechanism design for correlated bidders, two-sided markets, information structure design, etc. 
Additionally, another interesting open problem would be to develop a framework for proving lower bounds for mechanism design (e.g., lower bounds on the query complexity for single-round or multi-round protocols used to communicate with the bidders) using known limitations of algorithms in active learning, and vice-versa.

\subsection*{Acknowledgements}

Christos Boutsikas, Petros Drineas and Marios Mertzanidis are supported in part by a DOE award SC0022085, and NSF awards CCF-1814041, CCF-2209509, and DMS-2152687.
Alexandros Psomas and Paritosh Verma are supported in part by an NSF CAREER award CCF-2144208, a Google Research Scholar Award, and a Google AI for Social Good award.

\newpage
\bibliographystyle{alpha}
\bibliography{references.bib}

\newpage

\appendix
\section{Notation}

\begin{table}[ht!]
\caption{Notations}
\label{table: notation}
\centering
    \begin{tabular}{ p{3cm} p{8cm} }
    \toprule
 \textbf{Symbol} & \textbf{Explanation}\\
 \midrule
 \hline
 $n$ & Number of bidders\\
 \hline
 $m$ & Number of items\\
 \hline
 $k$ & Dimensionality of latent space\\
  \hline
$d$ & Number of types\\
 \hline
 $\D_i$ & Distribution of bidders $i$ type vector\\
 \hline
 $\Dhat_{i}$ & Induced noisy representation of $D_i$\\
  \hline
  $\Dhat_{z,i}$ & Distribution of latent vector type of bidder $i$\\
 \hline
  $\tb_i \in \mathbb{R}^d$ & Type vector of bidder $i$\\
 \hline
  $\zb_{i} \in \R^{k}$ & Sample from $\Dhat_{z,i}$ of bidder $i$ \\
  \hline
 $\Ab \in \mathbb{R}^{d \times k}$ & Design matrix of archetypes\\
 \hline
 $\Qcal(\tb_{i}) = \tilde{\zb}_{i} \in \R^{k}$ & Output of query protocol under true type $\tb_i$\\
 \hline
 $\epsb_{\texttt{mdl,p}} \in \R^{d}$ & Modeling vector noise \\
 \hline
 $ \emdl > 0$ & Upper bound on $\norm{\epsb_{\texttt{mdl,p}}}_{p}$ \\
 \hline
 $\epsb_{\texttt{nq,p}} \in \R^{d}$ & Query vector noise \\
 \hline
 $\enq \geq 0 $ & Upper bound on $\norm{\epsb_{\texttt{nq,p}}}_{p}$ \\
 \hline
 $v_i(\cdot,\cdot)$ & Valuation function\\
 \hline
 $\M$ & Mechanism\\
 \hline
 $x(\cdot)$ & Allocation rule\\
 \hline
 $p(\cdot)$ & Payment rule\\
 \hline
 $Rev(\M,\D)$ & Expected revenue of mechanism $M$ under type distributions $D$\\
 \hline
 $\pi_P(\cdot,\cdot)$ & Prokhorov distance of two probability measures\\
 \hline
 $\gamma$ & A coupling between two probability measures\\
 \hline
 $v^\Ab_i(\zb_i, S)$ & $v_i(\Ab \zb_i, S)$\\
 \hline
 $\Lcal$ & Lipschitz constant of valuation function\\
 \hline
 $\floor{F}_{\lambda, \delta}$ & Rounded down distribution\\
 \hline
 $r^{(\lambda, \delta)}(x)$ & Rounding down mapping function where $\delta$ is a parameter chosen by the designer\\
 \hline
 $\mathrm{supp}(F)$ & Support of probability distribution $F$\\
 \hline
 $\rho^{\ell}$ & Total TV distance of rounded down distributions\\
 \hline
 $u_i(\tb_i \leftarrow \bb_i, \tb_{-i})$ & Change in utility of bidder $i$ when bidder reports $\bb_i$ instead of $\tb_i$ and the remaining bidders bid $\tb_{-i}$\\
 \bottomrule
\end{tabular}
\end{table}

\section{Proof of Theorem~\ref{thm: combined results}}\label{app: combined proof}

\begin{proof}[Proof of~\cref{thm: combined results}]
For arbitrary valuation functions, $d=2^m$, and each type $\tb_i$ represents the valuation of bidder $i$ for all possible bundles of items. Since we can ask queries of the form ``What is your value for subset S?'', in \cref{linear-algebra-thm}, we have that $\varepsilon_{\texttt{nq,p}} = 0$. Thus from \cref{linear-algebra-thm}  with probability at least $1-\delta$ we have a $\left(\emdl, \dfrac{c_p \varepsilon_{\texttt{mdl},p}}{\sigma_{\min,p}(\Ab)}, p\right)$-query protocol that communicates once with the bidders and asks a total of $s_p$ queries. Where $s_1 = s_2 = O\left(k \cdot \ln k \cdot \ln \nicefrac{n}{\delta}\right)$ and $s_p = O\left(k^{\nicefrac{p}{2}}\cdot \ln^3 k \cdot \ln \nicefrac{n}{\delta}\right)$. Combining this observation with \cref{mainTheorem} proves the theorem. 
\end{proof}
\section{The Mechanism Design Component}\label{app:mechanism design}

The main goal of this section is to prove~\cref{mainTheorem}. In order to establish the theorem, we use  the following key lemma, which essentially establishes the robustness guarantees described in~\Cref{mainTheorem}, but in the space of latent types.

\begin{lemma} \label{mainMechanism}
Let $\Ab \in \RR{d}{k}$ be a design matrix. Suppose there exists a collection of distributions over latent types $\{\Fhat_{z,i}\}_{i \in [n]}$, where the support of each $\Fhat_{z,i}$ lies in $[0,1]^k$, and a mechanism $\Mhat$ that is IR and BIC w.r.t. $\Fhat_z = \times_{i=1}^{n} \Fhat_{z,i}$ and valuations $\{v^\Ab_i\}_{i \in [n]}$, where each $v_i$ is an $\Lcal$-Lipschitz valuation. Let $\Fcal_z = \times_{i=1}^n \Fcal_{z,i}$ be any distribution such that $\pi_p(\Fcal_{z,i}, \Fhat_{z,i}) \le \zeta_p$ for all $i \in [n]$. Given query access to $\Mhat$ (as defined in~\Cref{dfn: access to Mhat}), and sampling access to $\Fhat_{z}$ (as defined in~\Cref{dfn: sampling dhat}), we can construct a mechanism $\Mtilde$ such that the following hold: (i) $\Mtilde$ is IR and $(\eta, \mu)$-BIC w.r.t. $\Fcal_z$, and (ii) the expected revenue of $\Mtilde$ is $Rev(\Mtilde, \Fcal_z) \ge Rev(\Mhat, \Fhat_z) - O(n\eta)$, where $\mu = O(\zeta_p + \delta \, k^{1/p})$ and $\eta = O\left(k \Lcal \norm{\Ab}_{\infty}\left(\delta + n \, \left(1 + \frac{k^{1-1/p}}{\delta}\right)\, \zeta_p \right)\right)$, for all $\delta > 0$.
\end{lemma}

The proof of \cref{mainMechanism} is quite involved and will be the main focus of the subsequent analysis. Before proving \cref{mainMechanism}, we will show that it is sufficient to prove our main result, \cref{mainTheorem}.

\begin{proof}[Proof of \cref{mainTheorem}]
    Let $\tb_i$ be the type of bidder $i$, and $\zb_i$ be a random variable distributed according to $\Dhat_{z,i}$. We know that $\pi_p(\Dcal_i, \Ab \circ \Dhat_i) \le \emdl \le \zeta_p$. Due to~\cref{prokhorovCharacterization}, there exists a coupling such that with probability greater than $1-\zeta_p$, $\norm{\tb_i - \Ab \zb_i}_{p} \le \zeta_p$. Since the seller uses a $(\emdl, \zeta_p,p)$-query protocol, with probability at least $1-\zeta_p$, $\norm{\Qcal(\tb_i) - \zb_i}_{p} \le \zeta_p$. Note that, this implies that $\zb_i$ and $\Qcal(t_i)$ are distributed such that their Prokhorov distance is at most $\zeta_p$. 
    
    We now invoke~\cref{mainMechanism} by setting $\Fcal_{z,i}$ to be distributed as $\Qcal(\tb_i)$, $\Fhat_{z,i}$ to be distributed as $\Dhat_{z,i}$, and $\delta =\sqrt{\zeta_p}$. We run the resultant mechanism, $\Mtilde$, on types $\Qcal(\tb_1), \ldots, \Qcal(\tb_n)$, obtained by interacting with the bidders via the query protocol. With probability at least $1-\zeta_p$, we have that $\norm{\tb_i - \Ab \zb_i}_{\infty} \le \norm{\tb_i - \Ab \zb_i}_{p} \le \emdl$ and thus, using the fact that the query protocol ensures $\norm{\Qcal(\tb_i) - \zb_i}_{p} \le \zeta_p$ as well, we have $\norm{\tb_i - \Ab \Qcal(\tb_i)}_{\infty} \le \emdl + k\norm{\Ab}_{\infty}\zeta_p$. Note that $\Mtilde$ is a $(\eta, \mu)$-BIC mechanism w.r.t. $\times_{i \in [n]} \Fcal_{z,i}$, with $\eta = O(n\norm{\Ab}_{\infty} \Lcal \sqrt{\zeta_p})$ and $\mu = \sqrt{\zeta_p}$. Therefore, conditioned on the aforementioned scenario having probability at least $1-\zeta_p$, with probability at least $1-\mu$, deviating from interacting with $\Qcal$ truthfully can increase the expected utility by at most $O(\Lcal(\emdl + k\norm{\Ab}_{\infty}\zeta_p) + \eta) = O(\eta)$, and bidders can improve their utility by deviating with probability less than $\emdl \le \zeta_p$, i.e., with probability at most $1 - (1-\zeta_p)(1- \mu) = O(\zeta_p + \mu) = O(\mu)$. Hence, the overall mechanism is $(\eta, \mu)$-BIC w.r.t. $\D$, as well as IR (since $\Mtilde$ is IR). The revenue guarantee is immediate. This concludes the proof.   
\end{proof}




\subsection{Proof of Lemma~\ref{mainMechanism}}



For the remainder of this section, we use the following notation, where all distributions of an agent $i$ are supported on $[0,1]^k$: $\Fhat_z = \times_{i \in [n]}\Fhat_{z,i}$, $\floor{\Fhat_{z}}_{\ell, \delta} = \times_{i \in [n]} \floor{\Fhat_{z,i}}_{\ell, \delta}$, $\floor{\Fcal_{z}}_{\ell, \delta} = \times_{i \in [n]} \floor{\Fcal_{z,i}}_{\ell, \delta}$, and $\Fcal_z = \times_{i \in [n]}\Fcal_{z,i}$. We follow an approach similar to Brustle et al.~\cite{Brustle2020Robust}. The main idea is that arguing directly about mechanisms for distributions that are close in Prokhorov distance is difficult. On the flip side, arguing about mechanisms for distributions that are close in total variation distance is much easier, since the total variation is a more stringent --- and hence more well-behaved --- notion of distance. The key observation here is that if two distributions are close in Prokhorov distance then, in expectation, their rounded-down versions will also be close in total variation distance, where the expectation is taken over the randomness of parameter $\ell$. 

Our overall construction of the mechanism $\Mtilde$ (of~\cref{mainMechanism}) is via three reductions. For ease of presentation, we defer some proofs (and spefically, the proofs for Lemmas~\ref{roundingDown},~\ref{TVrobustness} and~\ref{RoundingUp}, to Appendix~\ref{appendix:mech design proof of lemmas}. First, in~\Cref{roundingDown}, given a mechanism for $\Fhat_z$ we design a mechanism for the rounded-down version.

\begin{restatable}{lemma}{roundingDown} \label{roundingDown}
Given query access to a mechanism $\Mcal$ that is IR and BIC w.r.t. $\Fhat_z$, and a sampling algorithm $\Scal_i$ for each $i \in [n]$, where $\Scal_i(\xb, \delta)$ draws a sample from the conditional distribution of $\Fhat_{z,i}$ on the $k$-dimensional cube $\times_{j \in [k]}[x_j, x_j+\delta_j)$, we can construct a mechanism $\Mcal_1^{\ell}$, oblivious to $\Fhat$, that is IR and $O(k \norm{\Ab}_{\infty} \Lcal \delta)$-BIC w.r.t. $\floor{\Fhat_{z}}_{\ell, \delta}$. Furthermore, \[ Rev\left(\Mcal_1^{\ell},  \floor{\Fhat_{z}}_{\ell, \delta}\right) \ge Rev(\Mcal, \Fhat_z) - O(kn \norm{\Ab}_{\infty} \Lcal \delta).\]
\end{restatable}

Second, in~\Cref{TVrobustness}, given a mechanism for $\floor{\Fhat_z}_{\ell,\delta}$ we design a mechanism for $\floor{\Fcal_z}_{\ell,\delta}$, which maintains its guarantees if $\pi_p(\Fcal_z,\Fhat_z)$ is small. 

\begin{restatable}{lemma}{TVrobustness} \label{TVrobustness}
Let $\floor{\Fhat_{z}}_{\ell, \delta}$  and $\floor{\Fcal_{z}}_{\ell, \delta}$ be distributions such that, for all $i\in [n]$, it holds that $(1)$ $\norm{\floor{\Fhat_{z,i}}_{\ell, \delta} - \floor{\Fcal_{z,i}}_{\ell, \delta}}_{TV} \le \varepsilon^{\ell}_i$, and $(2)$ $\pi_p\left(\Fhat_{z,i},  \Fcal_{z,i} \right) \le \zeta_p$.
Then, letting $\rho^{\ell} \coloneqq \sum_i \varepsilon_i^{\ell}$, given a mechanism $\Mcal^{\ell}_1$ that is IR and $\eta$-BIC w.r.t $\floor{\Fhat_{z}}_{\ell, \delta}$ we can construct a mechanism $\Mcal^{\ell}_2$ that is IR, $\left( O( k \Lcal \norm{\Ab}_{\infty}\rho^{\ell} + k\left(\zeta_p + \delta \cdot k^{1/p}\right)\norm{\Ab}_{\infty} \Lcal + \eta), \zeta_p + \delta k^{1/p}\right)$-BIC w.r.t. $\floor{\Fcal_{z}}_{\ell, \delta}$. Furthermore:
\[Rev\left(\Mcal_2^{\ell}, \floor{\Fcal_{z}}_{\ell, \delta}\right) \ge Rev\left(\Mcal_1^{\ell}, \floor{\Fhat_{z}}_{\ell, \delta}\right) -  nk\Lcal \norm{\Ab}_{\infty} \rho^{\ell} - nk(\zeta_p + \delta k ^{1/p})\norm{\Ab}_{\infty} \Lcal.\]
\end{restatable}

Third, in~\Cref{RoundingUp}, given a mechanism for $\floor{\Fcal_z}_{\ell,\delta}$ we design a mechanism for $\Fcal_z$.

\begin{restatable}{lemma}{RoundingUp} \label{RoundingUp}
    Given a mechanism $\M_2^{\ell}$ that is IR and $(\eta, \mu)$-BIC w.r.t. $\floor{\Fcal_z}_{\ell, \delta}$, we can construct a mechanism $\M^{\ell}$ that is IR and $\left( 3k\Lcal \norm{\Ab}_{\infty} \delta +\eta, \mu\right)$-BIC w.r.t. $\Fcal_z$. Moreover, $Rev\left(\M^{\ell}, \Fcal_z\right) \ge Rev\left(\M_2^{\ell}, \floor{\Fcal_z}_{\ell, \delta}\right) - nk\Lcal \norm{\Ab}_{\infty} \delta$.
\end{restatable}

With all the prerequisite technical lemmas at hand, we finally prove~\Cref{mainMechanism}. 




\begin{proof}[Proof of \cref{mainMechanism}]
The mechanism $\Mtilde$ simply operates as follows: $(i)$ $\ell$ is sampled uniformly from the interval $[0, \delta]$ (i.e., $\ell \sim \mathcal{U}[0,\delta]$) and $(ii)$ the mechanism $\M^{\ell}$ of~\Cref{RoundingUp} is run on the input bids. Specifically, $\Mtilde$ calls mechanism $\M^{\ell}$ of~\Cref{RoundingUp}, which calls mechanism $\M_2^{\ell}$ of~\Cref{TVrobustness}, which calls mechanism $\M_1^{\ell}$ of~\Cref{roundingDown}, which calls $\Mhat$ which is IR and BIC w.r.t. $\Fhat$. In order to invoke~\Cref{TVrobustness} we need a bound on the TV distance between $\floor{\Fhat_{z,i}}_{\ell, \delta}$ and $\floor{\Fcal_{z,i}}_{\ell, \delta}$. We use the following lemma (whose proof can be found in~\Cref{appendix:mech design proof of lemmas}).

\begin{restatable}{lemma}{TVdistanceBound}\label{TVdistanceBound}
Let $\Fcal_z$ and $\Fhat_z$ be two distributions supported on $\R^k$ such that $\pi_p(\Fcal_z, \Fhat_z) \le \varepsilon$. For any $\delta > 0$, $\EE_{\ell\sim \Ucal[0, \delta]^k} \left[\norm{\floor{\Fcal_z}_{\ell,\delta} - \floor{\Fhat_z}_{\ell,\delta}}_{TV}\right] \le \left(1+ \frac{k^{1-1/p}}{\delta} \right)\varepsilon$.
\end{restatable}

Letting $\rho^{\ell} = \sum_i \varepsilon^{\ell}_{i}$, where $\norm{\floor{\Fhat_{z,i}}_{\ell, \delta} - \floor{\Fcal_{z,i}}_{\ell, \delta}}_{TV} \le \varepsilon^{\ell}_i$,~\cref{TVdistanceBound} implies that $\EE_{\ell \sim \Ucal[0, \delta]}\left[\rho^{\ell} \right] \le n\left(1+ \frac{k^{1-1/p}}{\delta} \right)\zeta_p$, allowing us to invoke~\cref{TVrobustness}. The IR guarantee for $\Mtilde$ is immediate. For the BIC guarantee, there is a trade-off in choosing $\delta$: $\EE_{\ell \sim \Ucal[0, \delta]}\left[\rho^{\ell} \right]$ has a term inversely proportional to $\delta$, while the BIC guarantees of \cref{TVrobustness} and \cref{roundingDown} have a term proportional to $\delta$. We set $\delta = \sqrt{\zeta_p}$ to strike a balance between these terms. Combining the BIC and revenue guarantees of \cref{roundingDown}, \cref{TVrobustness}, and \cref{RoundingUp}, and using the fact that $k$ is a constant,  we get that $\Mtilde$ is a $(\eta, \mu)$-BIC and IR w.r.t. $\Fcal$ where $\mu = O(\zeta_p + \delta \cdot k^{1/p})$ and $\eta = O\left(k \Lcal \norm{\Ab}_{\infty}\left(\delta + n \cdot \left(1 + \frac{k^{1-1/p}}{\delta}\right)\cdot \zeta_p \right)\right)$ and $Rev(\Mtilde, \Fcal) \ge Rev(\Mhat, \Fhat) - O(n\eta)$.
\end{proof}

\subsection{Proofs of Lemma ~\ref{roundingDown}, \ref{TVrobustness}, \ref{RoundingUp}, and \ref{TVdistanceBound}}\label{appendix:mech design proof of lemmas}

The following proposition will be useful in multiple proofs throughout this section. 

\begin{proposition}\label{vALipschitzness}
If $v_i$ is a $\Lcal$-Lipschitz valuation function, then $v^\Ab_i$ is a $k \norm{\Ab}_{\infty}\Lcal$-Lipschitz valuation function.
\end{proposition}

\begin{proof}
The statement follows from the fact that $v^\Ab_i(\zb) = v_i(\Ab\zb)$, where $\zb \in [0,1]^k$ and $v_i$ is $\Lcal$-Lipschitz.
\end{proof}

\roundingDown*
\begin{proof}
 Let $x$ be the allocation rule and $p$ be the payment rule of $\Mcal$. We construct $\Mcal_1^{\ell}$ as follows. Upon receiving bids $\{\wb_i\}_{i \in [n]}$ we first query $\Scal_i$ for each bidder $i$ to get $\wb'_i = \Scal_i(\wb_i,\beta(\wb_i))$ where $\beta_j(\wb_i) = \delta$ if $w_{i,j} \neq 0$ and $\beta_j(\wb_i) = \ell_j$ otherwise. Then, each bidder $i$ gets allocated items $x_i(\wb')$ and pays $max\{0, p_i(\wb') - k \norm{\Ab}_{\infty} \Lcal \delta\}$.

 We first prove that this mechanism is IR. Note that, by definition of $\wb_i$, we have $\norm{\wb_i - \wb'_i}_{\infty} \le \delta$ for each $i \in [n]$. Since $\Mcal$ is IR then we know that for all bids $\tb_{-i}$ of other bidders,  $v^\Ab_i(\wb'_i, \Mcal(\wb'_i, \tb_{-i})) - p_i(\wb'_i, \tb_{-i}) = u_i(\wb'_i, \Mcal(\wb'_i, \tb_{-i})) \ge 0$. This gives us the required inequality since, 
 \begin{align*}
     0 & \leq v^\Ab_i(\wb'_i, \Mcal(\wb'_i, \tb_{-i})) - p_i(\wb'_i, \tb_{-i}) \\
     & = v^\Ab_i(\wb'_i, \Mcal(\wb'_i, \tb_{-i})) - k\norm{\Ab}_{\infty} \Lcal \delta - (p_i(\wb'_i, \tb_{-i}) - k\norm{\Ab}_{\infty} \Lcal \delta) \\
     & \leq v^\Ab_i(\wb_i, \Mcal(\wb'_i, \tb_{-i})). \tag{via \cref{vALipschitzness}}
 \end{align*}

 

 We now prove that the mechanism is $O(k \norm{\Ab}_{\infty} \Lcal \delta)$-BIC. From the point of view of bidder $i$ the types of the other bidders are drawn from $(\Fhat_z)_{-i}$ (i.e. $\wb'_{-i} \sim (\Fhat_z)_{-i}$). From the fact that $\Mcal$ is BIC w.r.t. $\Fhat_z$ we have the following:
 \begin{align*}\label{ineq:roundingdown0}
 \EE_{\wb'_{-i} \sim (\Fhat_z)_{-i}}\left[ u_i(\wb_i, \Mcal(\wb'_i, \wb'_{-i})) \right] &\ge \EE_{\wb'_{-i} \sim (\Fhat_z)_{-i}}\left[ u_i(\wb'_i, \Mcal(\wb'_i, \wb'_{-i}))\right] - k\norm{\Ab}_{\infty} \Lcal \delta \\
 &\geq \max_{\xb \in supp(\Fhat_{z,i})} \EE_{\wb'_{-i} \sim (\Fhat_z)_{-i}}\left[ u_i(\wb'_i, \Mcal(\xb, \wb'_{-i}))\right] - k\norm{\Ab}_{\infty} \Lcal \delta \\
 &\ge \max_{\xb \in supp(\Fhat_{z,i})} \EE_{\wb'_{-i} \sim (\Fhat_z)_{-i}}\left[ u_i(\wb_i, \Mcal(\xb, \wb'_{-i}))\right] - 2k\norm{\Ab}_{\infty} \Lcal \delta. \numberthis
 \end{align*}
 The first and the third inequalities above follow from \cref{vALipschitzness}. Thus if $i$ could pick the type exactly $\wb'_i$ as she pleased, she could not possibly make more than $2k\norm{\Ab}_{\infty} \Lcal \delta$. However she must pick a $\bb_i \in supp\left(\floor{\Fhat_{z,i}}_{\ell, \delta}\right)$, which gets rounded to $\bb'_i = \Scal_i(\bb_i, \beta(\bb_i))$. Specifically, $\bb'_i \sim \Fhat_{z,i}| \times_{j \in [k]} [b_{i,j}, b_{i,j} + \beta_j(\bb_i))$; for notational simplicity, we will simply denote this as $\bb'_i \sim \Scal_i(\bb_i, \beta(\bb_i))$. Bidder $i$'s utility when she reports bid $b_i \in supp\left(\floor{\Fhat_{z,i}}_{\ell, \delta}\right)$ can be bounded using the following observation.

\begin{align*}\label{ineq:roundingdown1}
 \EE_{\bb'_i \sim \Scal_i(\bb_i, \beta(\bb_i))}  \left[u_i(\wb_i, \Mcal(\bb'_i, \wb'_{-i})) \right]  = & \EE_{\bb'_i \sim \Scal_i(\bb_i, \beta(\bb_i))}  \left[ \EE_{\wb'_{-i} \sim {(\Fhat_z)_{-i}}}    \left[ u_i(\wb_i, \Mcal(\bb'_i, \wb'_{-i})) \right] \right] \\
  & \leq \max_{\xb \in supp(\Fhat_{z,i})}  \left[ \EE_{\wb'_{-i} \sim {(\Fhat_z)_{-i}}}\left[u_i(\wb_i, \Mcal(\xb, \wb'_{-i})) \right] \right]. \numberthis
\end{align*}

On combining inequalities (\ref{ineq:roundingdown0}) and (\ref{ineq:roundingdown1}), we get that our new mechanism is $2k \norm{\Ab}_{\infty} \Lcal \delta$-BIC w.r.t. $\times_{i \in [m]} \floor{\Fhat_{z,i}}_{\ell, \delta}$ .
Finally, note that under truthful bidding this mechanism extracts revenue at least $Rev(\Mcal,  \Fhat_z) - O(kn \norm{\Ab}_{\infty} \Lcal \delta)$ since $\wb'$ is essentially drawn from $\Fhat_z$ and each bidder gets a discount of $k \norm{\Ab}_{\infty} \Lcal \delta$.
\end{proof}

\TVrobustness*
\begin{proof}
    We will construct $\Mcal_2^{\ell}$ as follows. For every input bid $\wb \in supp\left( \times_{i \in [n]} \floor{\Fcal_{z,i}}_{\ell, \delta} \right)$, we first find, for each $i \in [n]$, the closest point in $\ell_p$ norm distance that is in $\floor{\Fhat_{z,i}}_{\ell, \delta}$; let $\wb_i' = \argmin_{x \in supp\left(  \floor{\Fhat_{z,i}}_{\ell, \delta} \right)} \norm{\wb-\xb}_{p}$ be this point. Notice here, that we assume that we can calculate the above expression exactly. However, according to~\cref{dfn: sampling dhat} we can actually compute $\hat{\wb}_i = \argmin_{x \in supp\left( \Fhat_{z,i}\right)} \norm{\wb-\xb}_{p}$. For the sake of simplicity, we continue the analysis as if we could calculate the desired expression. However, we can set $\wb_i' = r^{(\ell,\delta)}(\hat{\wb_i})$ (as defined in~\Cref{dfn: rounded}) and the following proposition implies that we only lose a small factor of $2\delta k^{\frac{1+p}{p}} \norm{\Ab}_{\infty} \Lcal$ in the BIC guarantee and a $2n\delta k^{\frac{1+p}{p}} \norm{\Ab}_{\infty} \Lcal$ factor in the revenue guarantee. In the following proposition and for the rest of our analysis we use the notation $\norm{\xb - \Acal}_p$ to denote the distance of a vector $\xb$ to the closest vector in a set of vectors $\Acal$, i.e., $\norm{\xb - \Acal}_p \coloneqq \min_{\yb \in \Acal} \norm{\xb - \yb}_p$.

    \begin{proposition}\label{roundedDownDIst}
    Let $\Bcal$ be a probability distributions supported on $[0,1]^{k}$. Then for any $\xb$ and $\wb = \argmin_{\zb \in supp\left(\Bcal\right)} \norm{\xb-\zb}_p$ we have that $\norm{\xb - r^{(\ell, \delta)}(\wb)}_p \le 2k^{1/p}\delta + \norm{\xb -supp\left(\floor{\Bcal}_{\ell, \delta}\right)}_p$.
    \end{proposition}
    \begin{proof}
    For any $\xb$ we have that $\norm{\xb - r^{(\ell, \delta)}(\xb)}_p \le k^{1/p} \delta$. Then, by using the triangle inequality for the $\ell_p$-norm and chaining the resulting inequalities, the proposition is implied.
    \end{proof}
    
    After finding $\wb'$, we then run $\Mcal_1^{\ell}$ on $\wb'$ giving a discount to each bidder, and at the same time making sure that the IR constraint is not violated. Let $x(\cdot)$ and $p(\cdot)$ be the allocation and payment rules of $\Mcal_1^{\ell}$. For each  $i \in [n]$ we do the following: if $\norm{w_i - w_i'}_p \le \left(\zeta_p + \delta \cdot k^{1/p}\right)$, bidder $i$ will receive allocation $x_i\left( \wb' \right)$ and pay $\max \{\widehat{p}_i(\wb'), 0\}$ where $\widehat{p}_i(\wb') \coloneqq p_i(\wb')- k\left(\zeta_p + \delta \cdot k^{1/p}\right)\norm{\Ab}_{\infty} \Lcal$. Otherwise, if  $\norm{w_i - w_i'}_p > \left(\zeta_p + \delta \cdot k^{1/p}\right)$, she will receive nothing and pays nothing. By construction, and using \cref{vALipschitzness}, mechanism $\Mcal_2^{\ell}$ is IR.

    Next, we show that $\Mcal_2^{\ell}$ is $\left( O( k \Lcal \norm{\Ab}_{\infty}\zeta_p + k\left(\zeta_p + \delta \cdot k^{1/p}\right)\norm{\Ab}_{\infty} \Lcal + \eta), \zeta_p + \delta k^{1/p}\right)$-BIC w.r.t. $\floor{\Fcal_{z}}_{\ell, \delta}$. As a first step, we prove that $\Mcal_2^{\ell}$ is $O(k\left(\zeta_p + \delta \cdot k^{1/p}\right)\norm{A}_{\infty} \Lcal + \eta)$-BIC w.r.t. $\floor{\Fhat_{z}}_{\ell, \delta}$. Recall that $\Mcal_1^\ell$ is $\eta$-BIC w.r.t. to $\floor{\Fhat_{z}}_{\ell, \delta}$, i.e., for all bidders $i$, latent types $\tb_i, \tb'_i$ we have,
    \begin{align*}
    & \EE_{\bb_{-i} \sim \floor{(\Fhat_z)_{-i}}_{\ell, \delta}} \left[ u_i(\tb_i,\Mcal_1^{\ell}(\tb_i,\bb_{-i})) \right] \ge \EE_{\bb_{-i} \sim \floor{(\Fhat_z)_{-i}}_{\ell, \delta}} \left[ u_i(\tb_i,\Mcal_1^{\ell}(\tb'_i,\bb_{-i})) \right] -\eta.
    \end{align*}
    This further implies that
    \begin{align*}\label{ineq:tv-0}
        & \EE_{\bb_{-i} \sim \floor{(\Fhat_z)_{-i}}_{\ell, \delta}}  \left[ u_i(\tb_i,\Mcal_1^{\ell}(\tb_i,\bb_{-i})) + k\left(\zeta_p + \delta \cdot k^{1/p}\right)\norm{\Ab}_{\infty} \Lcal\right] \\
    & \ge \EE_{\bb_{-i} \sim \floor{(\Fhat_z)_{-i}}_{\ell, \delta}} \left[ u_i(\tb_i,\Mcal_1^{\ell}(\tb'_i,\bb_{-i}))  + k\left(\zeta_p + \delta \cdot k^{1/p}\right)\norm{\Ab}_{\infty} \Lcal\right] - \eta \\
    & = \EE_{\bb_{-i} \sim \floor{(\Fhat_z)_{-i}}_{\ell, \delta}} \left[ v^\Ab_i(\tb_i,x_i(\tb'_i,\bb_{-i})) -  p_i(\tb'_i,\bb_{-i})+k\left(\zeta_p + \delta \cdot k^{1/p}\right)\norm{\Ab}_{\infty} \Lcal \right] - \eta \\
    & > \EE_{\bb_{-i} \sim \floor{(\Fhat_z)_{-i}}_{\ell, \delta}} \left[ v^\Ab_i(\tb_i,x_i(\tb'_i,\bb_{-i})) -  \max \{\widehat{p}_i(\tb'_i,\bb_{-i}), 0\} \right] -\eta. \numberthis
    \end{align*}
    Additionally, using the fact that $x = \max\{x,0\} + \min\{x,0\}$ for all $x \in \mathbb{R}$, we can upper bound the left-hand side to obtain,
    
    \begin{align*}\label{ineq:tv-1}
        &\EE_{\bb_{-i} \sim \floor{(\Fhat_z)_{-i}}_{\ell, \delta}}  \left[ v^\Ab_i(\tb_i,x_i(\tb_i,\bb_{-i})) -  p_i(\tb_i,\bb_{-i})+k\left(\zeta_p + \delta \cdot k^{1/p}\right)\norm{\Ab}_{\infty} \Lcal\right] \\
        & = \EE_{\bb_{-i} \sim \floor{(\Fhat_z)_{-i}}_{\ell, \delta}} \left[ v^\Ab_i(\tb_i,x_i(\tb_i,\bb_{-i})) - \max \{\widehat{p}_i(\tb'_i,\bb_{-i}), 0\} -  \min \{ \widehat{p}_i(\tb'_i,\bb_{-i}), 0\} \right] \\
    & \le \EE_{\bb_{-i} \sim \floor{(\Fhat_z)_{-i}}_{\ell, \delta}} \left[ v^\Ab_i(\tb_i,x_i(\tb_i,\bb_{-i})) - \max \{\widehat{p}_i(\tb'_i,\bb_{-i}), 0\} \right] + k\left(\zeta_p + \delta \cdot k^{1/p}\right)\norm{\Ab}_{\infty} \Lcal. \numberthis
    \end{align*}
    Combining inequalities (\ref{ineq:tv-0}) and (\ref{ineq:tv-1}), gives us the required inequality for the first step, i.e. that $\Mcal_2^{\ell}$ is $O(k\left(\zeta_p + \delta \cdot k^{1/p}\right)\norm{A}_{\infty} \Lcal + \eta)$-BIC w.r.t. $\floor{\Fhat_{z}}_{\ell, \delta}$.

    For the second step, we define, for types $\tb_i$, $\tb_{-i}$, and $\bb_i$, the following auxiliary function $u_i(\tb_i~\leftarrow \bb_i,~\tb_{-i}) \coloneqq u_i\left(\tb_i, \Mcal^\ell_2(\tb_i, \tb_{-i})\right) - u_i\left(\tb_i, \Mcal^\ell_2(\bb_i, \tb_{-i})\right)$. In other words $u_i(\tb_i \leftarrow \bb_i, \tb_{-i})$ simply represents the difference in utility of $i$ when he reports his true type $\tb_i$ as compared to $\bb_i$ to the mechanism $\Mcal^\ell_2$. Due to the Lipschitz continuity of valuation functions and the fact that $\Mcal^\ell_2$ is IR, we get that for all choices of $i, \tb_i, \bb_i, \tb_{-i}$ we must have $u_i(\tb_i, \Mcal^\ell_2(\bb_i, \tb_{-i})) \in [-k \Lcal \norm{\Ab}_{\infty}, k \Lcal \norm{\Ab}_{\infty}]$. Specifically, this follows from the following two observations: first, $v^\Ab_i(\mathbf{0}, \Mcal^\ell_2(\bb_i, \tb_{-i})) = 0$, thus due to Lipschitz continuity of $v^\Ab_i$ and the fact that $\tb_i \in [0,1]^k$, we have $v^\Ab_i(\tb, \Mcal^\ell_2(\bb_i, \tb_{-i})) \leq k \Lcal \norm{\Ab}_{\infty}$, and second, that the payments in $\Mcal^\ell_2$ are upper bounded by the maximum utility since $\Mcal^\ell_2$ is IR.  Therefore, $-2 k \Lcal \norm{\Ab}_{\infty} \le u_i(\tb_i \leftarrow \bb_i, \tb_{-i}) \leq 2 k \Lcal \norm{\Ab}_{\infty}$, or equivalently, for all $\xb_{-i}, \yb_{-i}$ we have that $\EE \left[u_i(\tb_i \leftarrow \bb_i, \xb_{-i}) \right]- \EE \left[u_i(\tb_i \leftarrow \bb_i, \yb_{-i}) \right] \le 4 k \Lcal \norm{\Ab}_{\infty} \Indic{\xb_{-i} \neq \yb_{-i}}$, where the expectation is taken over the randomness of the mechanism. This implies that for any coupling $\gamma$ of $\floor{(\Fhat_z)_{-i}}_{\ell, \delta},\floor{(\Fcal_z)_{-i}}_{\ell, \delta}$ --- and in particular for the coupling $\gamma^* = argmin_{\gamma} \EE_{(\xb_{-i},\yb_{-i}) \sim \gamma}\left[\Indic{\xb_{-i} \neq \yb_{-i}}\right]$ --- we have that, 
    \begin{align*}
        \EE_{(\xb_{-i},\yb_{-i}) \sim \gamma^*}\left[u_i(\tb_i \leftarrow \bb_i, \xb_{-i})- u_i(\tb_i \leftarrow \bb_i, \yb_{-i})\right] &\le 4 k \Lcal \norm{\Ab}_{\infty} \, \EE_{(\xb_{-i},\yb_{-i}) \sim \gamma^*}\left[\Indic{\xb_{-i} \neq \yb_{-i}}\right]\\
        &\le 4 k \Lcal \norm{\Ab}_{\infty} \, \norm{\floor{(\Fhat_z)_{-i}}_{\ell, \delta}-\floor{(\Fcal_z)_{-i}}_{\ell, \delta}}_{TV} \\
        &\le 4 k \Lcal \norm{\Ab}_{\infty} \, \rho^{\ell},
    \end{align*}
    where the final inequality follows from the fact that $\norm{ \floor{(\Fhat_z)_{-i}}_{\ell, \delta}- \floor{(\Fcal_z)_{-i}}_{\ell, \delta}}_{TV} \le \norm{ \floor{\Fhat_z}_{\ell, \delta}- \floor{\Fcal_z}_{\ell, \delta}}_{TV} \le \sum_i \varepsilon_i^{\ell} = \rho^{\ell}$. Note that the left hand side can be simplified as $\EE_{(\xb_{-i},\yb_{-i}) \sim \gamma^*}\left[u_i(\tb_i \leftarrow \bb_i, \xb_{-i})- u_i(\tb_i \leftarrow \bb_i, \yb_{-i})\right] = \EE_{\xb_{-i} \sim \floor{(\Fhat_z)_{-i}}_{\ell, \delta}}\left[u_i(\tb_i \leftarrow \bb_i, \xb_{-i})\right]-$\\$\EE_{\yb_{-i} \sim \floor{(\Fcal_z)_{-i}}_{\ell, \delta}}\left[u_i(\tb_i \leftarrow \bb_i, \yb_{-i})\right]$, and therefore, rearranging give us,
    \begin{align*} \label{ineq:bic-guarantee}
        \EE_{\yb_{-i} \sim \floor{(\Fcal_z)_{-i}}_{\ell, \delta}} & \left[u_i\left(\tb_i, \Mcal^\ell_2(\tb_i, \yb_{-i})\right)\right] - \EE_{\yb_{-i} \sim \floor{(\Fcal_z)_{-i}}_{\ell, \delta}}\left[u_i\left(\tb_i,\Mcal^\ell_2(\bb_i, \yb_{-i})\right)\right] \\
        & = \EE_{\yb_{-i} \sim \floor{(\Fcal_z)_{-i}}_{\ell, \delta}}\left[u_i(\tb_i \leftarrow \bb_i, \yb_{-i})\right] \\
        & \ge  \EE_{\xb_{-i} \sim \floor{(\Fhat_z)_{-i}}_{\ell, \delta}}\left[u_i(\tb_i \leftarrow \bb_i, \xb_{-i})\right] - 4 k \Lcal \norm{\Ab}_{\infty} \rho^{\ell}. \numberthis
    \end{align*}
    
    As the final step of establishing the BIC guarantee of $\Mcal_2^{\ell}$, we lower bound $\EE_{\xb_{-i} \sim \floor{(\Fhat_z)_{-i}}_{\ell, \delta}}\left[u_i(\tb_i \leftarrow \bb_i, \xb_{-i})\right]$. Note that $\Mcal_2^{\ell}$ maps the reported bids to the closest point in $supp\left(\floor{\Fhat_z}_{\ell, \delta} \right)$, and hence, without loss of generality we can assume that when bidders misreport they choose a bid $\bb_i \in supp\left(\floor{\Fhat_z}_{\ell, \delta} \right)$. Letting $\tb^*_i = \argmin_{\xb \in supp\left(  \floor{\Fhat_{z,i}}_{\ell, \delta} \right)} \norm{\tb_i-\xb}_{p}$ and $\lambda = \norm{\tb_i - supp\left(\floor{\Fhat_z}_{\ell, \delta} \right)}_p$, we have that, 
    \begin{align*}
        &\EE_{\xb_{-i} \sim \floor{(\Fhat_z)_{-i}}_{\ell, \delta}}\left[u_i(\tb_i \leftarrow \bb_i, \xb_{-i})\right] = \EE_{\xb_{-i} \sim \floor{(\Fhat_z)_{-i}}_{\ell, \delta}}\left[u_i\left(\tb_i,\Mcal^\ell_2(\tb_i, \xb_{-i})\right) - u_i\left(\tb_i, \Mcal^\ell_2(\bb_i, \xb_{-i})\right)\right] \\
        &\quad= \EE_{\xb_{-i} \sim \floor{(\Fhat_z)_{-i}}_{\ell, \delta}}\left[u_i\left(\tb_i,\Mcal^\ell_2(\tb^*_i, \xb_{-i})\right) - u_i\left(\tb_i, \Mcal^\ell_2(\bb_i, \xb_{-i})\right)\right] \\
        &\quad \ge \EE_{\xb_{-i} \sim \floor{(\Fhat_z)_{-i}}_{\ell, \delta}}\left[u_i\left(\tb^*_i,\Mcal^\ell_2(\tb^*_i, \xb_{-i})\right) - u_i\left(\tb^*_i, \Mcal^\ell_2(\bb_i, \xb_{-i})\right)\right] - 2k \Lcal \norm{\Ab}_{\infty}\lambda.
    \end{align*}

    To complete the argument, we use the following two propositions. 

\begin{proposition}\label{roundedProkhoroveDistance}
    Let $\Bcal$ and $\Bcal'$ be two probability distributions supported on $[0,1]^{k}$ for which $\pi_p(\Bcal,\Bcal') \le \varepsilon$. Then it must be true that $\pi_p\left(\floor{\Bcal}_{\ell, \delta},  \floor{\Bcal'}_{\ell, \delta} \right) \le \varepsilon + \delta \cdot k^{1/p}$.
\end{proposition}

\begin{proof}
    Intuitively, the proposition follows since two probability masses, when rounded, will move at most by an additive factor of $\delta \cdot k^{1/p}$ in the $\ell_p$ norm distance compared to each other. Formally, since $\pi_{P}(\Bcal,\Bcal') \le \varepsilon$ from \cref{prokhorovCharacterization} we know that there exist coupling $\gamma$ of $\Bcal, \Bcal'$ such that $\PP_{(x,y) \sim \gamma}\left[\norm{x-y}_p > \varepsilon\right] \le \varepsilon$. From the definition of rounding, we know that $\norm{r^{(\ell, \delta)}(x)-r^{(\ell, \delta)}(y)}_p \le \norm{x-y + \overrightarrow{\delta}}_p \le \norm{x-y}_p + \delta \norm{\overrightarrow{1}}_p = \norm{x-y}_p + \delta \cdot k^{1/p}$. Therefore, for the same coupling $\gamma$ we have that, 

    \begin{align*}
    \PP_{(x,y) \sim \gamma} \left[ \norm{r^{(\ell, \delta)}(x)-r^{(\ell, \delta)}(y)}_p > \varepsilon + \delta \cdot k^{1/p}\right] & \le \PP_{(x,y) \sim \gamma}\left[ \norm{x-y}_p + \delta \cdot k^{1/p} > \varepsilon + \delta \cdot k^{1/p}\right] \\
    & \le \varepsilon \le \varepsilon + \delta \cdot k^{1/p}. 
     \end{align*}
\end{proof}

\begin{proposition}\label{prokhoroveSupportDistance}
    Let $\Bcal$ and $\Bcal'$ be two probability distributions supported on $[0,1]^{k}$ for which $\pi_p(\Bcal,\Bcal') \le \varepsilon$. Then it must be true that $\PP_{x \sim \Bcal}\left[\norm{\xb- supp(\Bcal')}_p > \varepsilon\right] \le \varepsilon$.
\end{proposition}

\begin{proof}
    Using \cref{prokhorovCharacterization} we know that there exist coupling $\gamma$ of $\Bcal, \Bcal'$ such that $\PP_{(x,y) \sim \gamma}\left[\norm{x-y}_p > \varepsilon\right] \le \varepsilon$. This directly gives us the required inequality, $\PP_{x \sim B}\left[\norm{x- supp(\Bcal')}_p > \varepsilon\right] \le \PP_{x \sim \Bcal, y \sim \gamma|x}\left[\norm{x - y}_p > \varepsilon\right] = \PP_{(x,y) \sim \gamma}\left[\norm{x-y}_p > \varepsilon\right] \le \varepsilon$.
\end{proof}
    
    First, as previously shown, our mechanism $\Mcal^\ell_2$ is a $k\left(\zeta_p + \delta \cdot k^{1/p}\right)\norm{\Ab}_{\infty} \Lcal + \eta$-BIC w.r.t. $\times_{i \in [n]} \floor{\Fhat_{z,i}}_{\ell, \delta}$. Therefore, for every $i, \tb_i$, and $\bb_i$ we have $\EE_{\tb_{-i} \sim \floor{(\Fhat_z)_{-i}}_{\ell, \delta}}\left[u_i(\tb_i \leftarrow \bb_i, \tb_{-i})\right] \ge -k\left(\zeta_p + \delta \cdot k^{1/p}\right)\norm{\Ab}_{\infty} \Lcal - \eta$. 
    Second, we know that $\pi_p\left(\Fhat_{z,i},  \Fcal_{z,i} \right) \le \zeta_p$, and hence via \cref{roundedProkhoroveDistance}, we have $\pi_p\left(\floor{\Fhat_{z,i}}_{\ell, \delta},  \floor{\Fcal_{z,i}}_{\ell, \delta} \right) \le \zeta_p + \delta \cdot k^{1/p}$. Hence, by invoking \cref{prokhoroveSupportDistance}, we get that $\lambda = \norm{\tb_i - supp\left(\floor{\Fhat_z}_{\ell, \delta} \right)}_p \le \zeta_p + \delta \cdot k^{1/p}$ with probability at least $1- \zeta_p - \delta \cdot k^{1/p}$. Using these two observations we get that,
    \begin{align*}
        \EE_{\tb_{-i} \sim \floor{(\Fhat_z)_{-i}}_{\ell, \delta}}\left[u_i(\tb_i \leftarrow \bb_i, \tb_{-i})\right] & \ge  -k\left(\zeta_p + \delta \cdot k^{1/p}\right)\norm{\Ab}_{\infty} \Lcal - \eta - 2 k \lambda \norm{\Ab}_{\infty} \Lcal \\
        & \geq -3k\left(\zeta_p + \delta \cdot k^{1/p}\right)\norm{\Ab}_{\infty} \Lcal - \eta .
    \end{align*} 
    On combining this with \cref{ineq:bic-guarantee}, we get that with probability at least $1- \zeta_p - \delta \cdot k^{1/p}$,

    \begin{align*}
        \EE_{\xb_{-i} \sim \floor{(\Fcal_z)_{-i}}_{\ell, \delta}} & \left[u_i\left(\tb_i, \Mcal^\ell_2(\tb_i, \xb_{-i})\right)\right] - \EE_{\xb_{-i} \sim \floor{(\Fcal_z)_{-i}}_{\ell, \delta}}\left[u_i\left(\tb_i,\Mcal^\ell_2(\bb_i, \xb_{-i})\right)\right] \\
        & \ge -3k\left(\zeta_p + \delta \cdot k^{1/p}\right)\norm{\Ab}_{\infty} \Lcal - k \Lcal \norm{\Ab}_{\infty} \rho^{\ell} - \eta.
    \end{align*}
    That is, $\Mcal^\ell_2$ is $\left( O( k \Lcal \norm{\Ab}_{\infty}\rho^{\ell} + k\left(\zeta_p + \delta \cdot k^{1/p}\right)\norm{\Ab}_{\infty} \Lcal + \eta), \zeta_p + \delta k^{1/p}\right)$-BIC w.r.t. $\floor{\Fcal_{z}}_{\ell, \delta}$.

    Finally, we prove the revenue guarantee of $\Mcal^\ell_2$. Since our mechanism operates by offering a discount of $k\left(\zeta_p + \delta \cdot k^{1/p}\right)\norm{\Ab}_{\infty} \Lcal$ to each bidder, we have $Rev\left(\Mcal_2^{\ell}, \floor{\Fhat_{z}}_{\ell, \delta}\right) \ge Rev\left(\Mcal_1^{\ell}, \floor{\Fhat_{z}}_{\ell, \delta}\right) - nk(\zeta_p + \delta k ^{1/p})\norm{\Ab}_{\infty} \Lcal$. Next, we consider the following two cases. Let $\tb \sim \floor{\Fcal_z}_{\ell, \delta}$ and $\hat{\tb} \sim \floor{\Fhat_z}_{\ell, \delta}$. If $\tb = \hat{\tb}$, mechanism $\Mcal_2^{\ell}$ will produce the same revenue. Otherwise, if $\tb \neq \hat{\tb}$ one instance can extract at most $nk\Lcal \norm{\Ab}_{\infty}$ more revenue than the other. Since $\norm{ \floor{\Fhat}_{\ell, \delta}- \floor{\Fcal}_{\ell, \delta}}_{TV} \le \sum_i \varepsilon_i^{\ell} = \rho^{\ell}$, from the characterization of TV distance~\cite{LevinPeresWilmer2009}\footnote{That is, $\norm{X-Y}_{TV} = \min_{\gamma} \PP_{(X,Y)\sim \gamma}\left[ X \neq Y \right]$, where $\gamma$ is the minimum over all couplings of $X$ and $Y$.} we know that there exists a coupling such that $\tb \neq \hat{\tb}$ with probability less than $\rho^{\ell}$. This gives us the desired bound on revenue:
    \begin{align*}
        Rev&\left(\Mcal_2^{\ell}, \floor{\Fcal_{z}}_{\ell, \delta}\right) \ge Rev\left(\Mcal_2^{\ell}, \floor{\Fhat_{z}}_{\ell, \delta}\right) - nk\Lcal \norm{\Ab}_{\infty} \rho^{\ell} \\
        &\ge Rev\left(\Mcal_1^{\ell}, \floor{\Fhat_{z}}_{\ell, \delta}\right) - nk\Lcal \norm{\Ab}_{\infty} \rho^{\ell} - nk (\zeta_p + \delta k ^{1/p})\norm{\Ab}_{\infty} \Lcal. \qedhere
    \end{align*}
\end{proof}

\RoundingUp*
\begin{proof}
    Let $x$ be the allocation rule and $p$ be the payment rule of $\Mcal_2^{\ell}$. The mechanism $\Mcal^{\ell}$ operates as follows: given the input bid $\wb_i$ of each  $i$, we first construct $\wb'_i = r^{(\ell, \delta)}(\wb_i)$, and then we allocate to bidder $i$, the items $x_i(\wb')$ and make him pay $\max\{p_i(\wb')-k\norm{\Ab}_{\infty}\Lcal\delta, 0\}$. Note that if $\wb_i \sim \Fcal_{z,i}$ then $\wb'_i \sim \floor{\Fcal_{z,i}}_{\ell, \delta}$. 

    We first argue that the mechanism is $\left(3k\Lcal \norm{\Ab}_{\infty} \delta +\eta, \mu \right)$-BIC wrt $\Fcal_z$.

    \begin{align*}
        \EE_{\tb_{-i} \sim (\Fcal_z)_{-i}}\left[ u_i(\wb_i, \Mcal^{\ell}(\wb_i, \tb_{-i}))\right] & \ge \EE_{\tb'_{-i} \sim \floor{(\Fcal_z)_{-i}}_{\ell, \delta}}\left[ u_i(\wb_i, \Mcal_2^{\ell}(\wb'_i, \tb'_{-i})) \right] \\
    & \ge \EE_{\tb'_{-i} \sim \floor{(\Fcal_z)_{-i}}_{\ell, \delta}}\left[ u_i(\wb'_i, \Mcal_2^{\ell}(\wb'_i, \tb'_{-i})) \right] - k\norm{\Ab}_{\infty} \Lcal \delta.
    \end{align*}

    The first inequality follows from the definition of the mechanism $\Mcal_2^{\ell}$ and the last inequality follows from the fact that $\norm{\wb'_i -\wb_i}_{\infty} \le \delta$. Towards completing the proof of the BIC guarantee, we will now lower bound the right-hand side. To this end, note that the mechanism $\Mcal_2^{\ell}$ is $(\eta, \mu)$-BIC, i.e., with probability at least $1-\mu$, for any $\bb_i \in supp(\Fcal_{z,i})$ we have,

    \begin{align*}
    \EE_{\tb'_{-i} \sim \floor{(\Fcal_z)_{-i}}_{\ell, \delta}} & \left[ u_i(\wb'_i, \Mcal_2^{\ell}(\wb'_i, \tb'_{-i})) \right] - k\norm{\Ab}_{\infty} \Lcal \delta \\
    & \ge \EE_{\tb'_{-i} \sim \floor{(\Fcal_z)_{-i}}_{\ell, \delta}}\left[ u_i(\wb'_i, \Mcal_2^{\ell}(r^{(\ell, \delta)}(\bb_i), \tb'_{-i})) \right] - \eta - k\norm{\Ab}_{\infty} \Lcal \delta \\
    & \ge \EE_{\tb'_{-i} \sim \floor{(\Fcal_z)_{-i}}_{\ell, \delta}}\left[ u_i(\wb_i, \Mcal_2^{\ell}(r^{(\ell, \delta)}(\bb_i), \tb'_{-i})) \right] - \eta - 2k\norm{\Ab}_{\infty} \Lcal \delta \\
    & \ge \EE_{\tb_{-i} \sim (\Fcal_z)_{-i}}\left[ u_i(\wb_i, \Mcal^{\ell}(r^{(\ell, \delta)}(\bb_i), \tb_{-i})) \right] -\eta- 3k\norm{\Ab}_{\infty} \Lcal \delta\\
    & = \EE_{\tb_{-i} \sim (\Fcal_z)_{-i}}\left[ u_i(\wb_i, \Mcal^{\ell}(\bb_i, \tb_{-i})) \right] -\eta- 3k\norm{\Ab}_{\infty} \Lcal \delta,
    \end{align*}

     The second inequality follows from the fact that $\norm{\wb'_i -\wb_i}_{\infty} \le \delta$ and the last is due to the definition of the mechanism $\Mcal^{\ell}$. 

    We now argue that the mechanism $\Mcal^{\ell}$ is IR. To establish this, we only need to check instances where $p_i(\wb') \neq 0$; if $p_i(\wb') = 0$, then the mechanism is trivially IR. Using the fact that $\Mcal_2^{\ell}$ is IR and $\norm{\wb_i - \wb'_i}_{p} \le \delta$, we have that

    \begin{align*}
        0 \le u_i\left(\wb'_i, \left(\Mcal_2^{\ell}(\wb'_i, \wb'_{-i})\right)\right) & = v^\Ab_i\left(\wb'_i, \left(\Mcal_2^{\ell}(\wb'_i, \wb'_{-i})\right)\right) - p_i(\wb'_i, \wb'_{-i})\\
    & \le v^\Ab_i\left(\wb_i, \left(\Mcal_2^{\ell}(\wb'_i, \wb'_{-i})\right)\right) - p_i(\wb'_i, \wb'_{-i}) + k\norm{\Ab}_{\infty}\Lcal\delta \\
    & = u_i\left(\wb_i,  \left(\Mcal^{\ell}(\wb_i, \wb'_{-i})\right)\right).
    \end{align*}

    Therefore, $\Mcal^{\ell}$ is IR. Finally, note that when all bidders bid truthfully $\Mcal^{\ell}$ extracts the same revenue as $\Mcal_2^{\ell}$ with a cumulative discount of at most $nk\Lcal \norm{\Ab}_{\infty} \delta$ for all the bidders, this directly implies the revenue bound stated in the lemma statement.
\end{proof}

\TVdistanceBound*
\begin{proof}
Using \cref{prokhorovCharacterization} we know that there exists a coupling $\gamma$ of $\Fcal_z$ and $\Fhat_z$ so that $\PP_{(x,y) \sim \gamma}[\norm{x-y}_p > \varepsilon] \le \varepsilon$. Thus, we can bound the following probability:

\begin{align*}
\PP_{\ell \sim \Ucal[0, \delta]^k, (x,y) \sim \gamma} \left[ r^{(\ell, \delta)}(x) \neq r^{(\ell, \delta)}(y) \right]  = 
  \PP_{\ell \sim \Ucal[0, \delta]^k, (x,y) \sim \gamma} \left[ r^{(\ell, \delta)}(x) \neq r^{(\ell, \delta)}(y) \wedge \norm{x-y}_p > \varepsilon \right]+\\ 
 \PP_{\ell \sim \Ucal[0, \delta]^k, (x,y) \sim \gamma} \left[ r^{(\ell, \delta)}(x) \neq r^{(\ell, \delta)}(y) \wedge \norm{x-y}_p \le \varepsilon \right]. 
\end{align*}
We can upper bound the first term above by $\PP_{(x,y) \sim \gamma}\left[ \norm{x-y}_p > \varepsilon \right] \leq \varepsilon$, and using Bayes rule, the second term can be written as:

\[\PP_{\ell \sim \Ucal[0, \delta]^k} \left[ r^{(\ell, \delta)}(x) \neq r^{(\ell, \delta)}(y) \; | \; \norm{x-y}_p \le \varepsilon \right] \cdot \PP_{(x,y) \sim \gamma}\left[ \norm{x-y}_p \le \varepsilon \right] \leq \PP_{\ell \sim \Ucal[0, \delta]^k} \left[ r^{(\ell, \delta)}(x) \neq r^{(\ell, \delta)}(y) \; | \; \norm{x-y}_p \le \varepsilon \right]\]

Plugging these upper bounds we get,

\begin{align*}
\PP_{\ell \sim \Ucal[0, \delta]^k, (x,y) \sim \gamma} \left[ r^{(\ell, \delta)}(x) \neq r^{(\ell, \delta)}(y) \right] & \le
 \varepsilon + \PP_{\ell \sim \Ucal[0, \delta]^k} \left[ r^{(\ell, \delta)}(x) \neq r^{(\ell, \delta)}(y) \; | \; \norm{x-y}_p \le \varepsilon \right] \\
& \le \varepsilon + \sum_{i \in [k]}\PP_{\ell_i \sim \Ucal[0, \delta]} \left[ r_i^{(\ell, \delta)}(x) \neq r_i^{(\ell, \delta)}(y) \; | \; \norm{x-y}_p \le \varepsilon \right]\\
& \le \left(1 + \frac{k^{1-1/p}}{\delta}\right)\varepsilon.
\end{align*}

The final inequality follows from the fact that $\sum_{i \in [k]}\PP_{\ell_i \sim \Ucal[0, \delta]} \left[ r_i^{(\ell, \delta)}(x) \neq r_i^{(\ell, \delta)}(y) \right] \le \sum_{i \in [k]} \frac{|x_i - y_i|}{\delta} = \frac{\norm{x-y}_1}{\delta} \le \frac{k^{1-1/p}\norm{x-y}_p}{\delta}$, where in the final inequality here we used Holder's inequality.
\end{proof}

\section{Active learning for regression problems via Randomized Linear Algebra: Details}\label{sec:appla1}

\subsection{Our query protocol: details for $p=1$ and $p=2$}\label{appendix:linear-algebra:query12}

Our query protocol takes as input the archetype matrix $\Ab$ and a set of sampling probabilities $q_i$, $i \in [d]$, summing up to one. It outputs a sampling (and a rescaling) matrix  that can be used to select a small subset of types to query bidders' preferences.

\begin{algorithm}[H] 
         \SetAlgoLined
	     \KwIn{$\Ab \in \R^{d \times k}$, sampling complexity $s_{p}$, probabilities $q_i>0$, $i \in [d]$, $\sum_{i=1}^d q_i = 1$}
	     \KwOut{sampling matrix $\Sb_p \in \R^{s_p \times d}$, rescaling matrix $\Db_p \in \R^{s_p \times s_p}$} \vspace*{0.2cm}
          Initialize $\Sb_p, \, \Db_p$ to be all-zero matrices;\\
	   \For{$t\gets1$  \KwTo $s_p$ }{
	       Sample index $j \in [d]$ with respect to the probabilities $q_{1}\ldots q_{d}$;\\
                 $\Sb_{p}(t,j) \gets 1$ \tcp*{Set the $t$th row of $\Sb$ to $\eb_{j}$}
                $\Db_{p}(t,t) \gets \frac{1}{\sqrt{s_{p} q_{j}}}$ \tcp*{and rescale}
	   }
        \caption{Sampling \& Rescaling Algorithm}
        \label{al1}
\end{algorithm}

\subsubsection{The $p=2$ case}

We start with the $p=2$ case and note that $\ell_2$ regression is the most studied problem in the RLA literature, including the active learning setting. In this case, the sampling probabilities will be the so-called (row) leverage scores of $\Ab$, which can be computed exactly in $O(dk^2)$ time. Leverage scores can be approximated faster, in time that depends basically on the sparsity of the input matrix and we refer the reader to~\cite{Mah-mat-rev_BOOK,Woodruff2014,drineas2016randnla} for details on this very well-studied quantities. The main quality-of-approximation result is captured in the following lemma.

\begin{lemma} \label{lem:pequalone_p_2}
    Let $\Ab \in \R^{d \times k}$ and $(\tb + \epsb_{\texttt{nq,2}}) \in \R^{d}$. Assume that the sampling probabilities $q_i$ of Algorithm~\ref{al1} are the row leverage scores of $\Ab$. 
    Let $\tilde{\zb} \in \mathbb{R}^k$ be 
    $$\tilde{\zb} = \arg \min_{\zb \in  \R^{k}}\norm{{\Db_{2} \Sb_{2} \Ab \zb - \Db_{2} \Sb_{2}  (\tb + \epsb_{\texttt{nq,2}})}}_{2}.$$ 
    Then, with probability at least $0.99$,
    \begin{equation*}
        \norm{\Ab \tilde{\zb} - (\tb + \epsb_{\texttt{nq,p}})}_{p} \leq \gamma_2 \mathtt{OPT},
    \end{equation*}
    where $\mathtt{OPT} = \min_{\zb \in \R^{k}}\norm{\Ab \zb - (\tb + \epsb_{\texttt{nq,2}})}_{2}$. Here $\gamma_2 = 1+\erla$, where $\erla > 0$ and the query complexity $s'_2$ satisfies
    $$s'_2 = O(k\ln k + \nicefrac{k}{\erla}).$$
\end{lemma}
The proof of the above lemma follows from Lemmas 4 and 5 of~\cite{drineas2011faster}, which each hold with probability at least $1-\nicefrac{\delta}{2}$, by setting the query complexity to $s_2 = O(k \log k + \nicefrac{k}{\erla})$ and using the leverage scores as sampling probabilities. Applying a union bound over the failure probabilities of the two lemmas concludes the proof. 

We now use Theorem~3.3 of~\cite{musco2021active_arxiv} to boost the success probability of the above algorithm. 
Specifically, we repeat our algorithm $O\left(\ln \nicefrac{1}{\delta'} \right)$ times to derive multiple candidate solutions for any $\delta' \in (0,1)$. Then, we can use Algorithm 3 of~\cite{musco2021active_arxiv} to select a solution that satisfies a slightly worse accuracy guarantee with high probability. More precisely, our final solution $\tilde{\zb}$ will satisfy
\begin{equation}\label{eqn:pdappendix_pequal1}
    \norm{\Ab \tilde{\zb} - (\tb + \epsb_{\texttt{nq,p}})}_{p} \leq (3\cdot (1+\erla)+2) \mathtt{OPT}
\end{equation}
with probability at least $1-\delta'$. Fixing $\erla = 0.5$, we get an overall query complexity equal to
$$s_2 = O\left( \ln\left( \nicefrac{1}{\delta'} \right) s'_{2} \right) = O\left( k \ln k \ln\left( \nicefrac{1}{\delta'} \right) \right).$$ 
Eventually, we will need the bound of eqn.~(\ref{eqn:pdappendix_pequal1}) to hold for all $n$ bidders via a union bound, so the failure probability must be reduced to $\delta' = \nicefrac{\delta}{n}$. Recall that $\erla = 0.5$; eqn.~(\ref{eqn:pdappendix_pequal1}) becomes 
\begin{equation}\label{eqn:pdappendix_pequal1f}
    \norm{\Ab \tilde{\zb} - (\tb + \epsb_{\texttt{nq,p}})}_{p} \leq 6.5\cdot \mathtt{OPT}.
\end{equation}
The above bound holds with probability at least $1-\delta' = 1-\nicefrac{\delta}{n}$ if 
\begin{equation*}
    s_{2} = O\left( k \ln k  \ln\left( \nicefrac{n}{\delta} \right) \right),
\end{equation*}
for any $\delta \in (0,1)$.

\subsubsection{The $p=1$ case}

We now focus on the $p=1$ case. In this setting, the sampling probabilities will be the Lewis weights, which can be approximated efficiently following the lines of~\cite{cohen2015lp, chen2021query}. We do emphasize that approximations to the Lewis weights (or the leverage scores) are sufficient in our setting. We now directly apply Theorem 1.2 of~\cite{chen2021query}). We again need a failure probability that is at most $\delta/n$, since we will need to apply a union bound over all $n$ bidders. In our notation, Theorem 1.2 of~\cite{chen2021query} can be restated as follows.

\begin{lemma} \label{lem:pequalone_p_1}
    Let $\Ab \in \R^{d \times k}$ and $(\tb + \epsb_{\texttt{nq,1}}) \in \R^{d}$. Assume that the sampling probabilities $q_i$ of Algorithm~\ref{al1} are the row Lewis weights ($p=1$) of $\Ab$. 
    Let $\tilde{\zb} \in \mathbb{R}^k$ be 
    $$\tilde{\zb} = \arg \min_{\zb \in  \R^{k}}\norm{{\Db_{1} \Sb_{1} \Ab \zb - \Db_{1} \Sb_{1}  (\tb + \epsb_{\texttt{nq,1}})}}_{1}.$$ 
    Then, with probability at least $1-(\delta/n)$,
    \begin{equation*}
        \norm{\Ab \tilde{\zb} - (\tb + \epsb_{\texttt{nq,1}})}_{1} \leq \gamma_1 \mathtt{OPT},
    \end{equation*}
    where $\mathtt{OPT} = \min_{\zb \in \R^{k}}\norm{\Ab \zb - (\tb + \epsb_{\texttt{nq,1}})}_{1}$. Here $\gamma_1 = 1+\erla$, where $\erla > 0$ and the query complexity $s_1$ satisfies
    $$s_1 = O\left(\nicefrac{k}{\erla^2} \ln \nicefrac{kn}{\erla \delta}\right).$$
\end{lemma}
In our proof of Theorem~\ref{linear-algebra-thm}, we will set $\erla$ to 0.5 for simplicity.

\subsection{Our query protocol: details for $p>2$}\label{appendix:linear-algebra:querylargerthan2}
In this section, we show how to adapt the results of~\cite{musco2022active, musco2021active_arxiv} in our setting in order to prove Theorem~\ref{linear-algebra-thm}). We restate a sequence of results from~\cite{musco2022active, musco2021active_arxiv}, focusing on $p >2$ and using our notation.

\begin{lemma}[Theorem 2.11 in \cite{musco2021active_arxiv}] \label{musco_1}
    Let $3 \leq p < \infty$ be an integer. There exists a randomized algorithm which constructs a sampling matrix $\Sb_{p} \in \R^{s'_{p} \times d}$ and a rescaling matrix $\Db_{p} \in \R^{s'_{p} \times s'_{p}}$ such that, with probability at least $0.99$, the $\ell_{p}$ subspace embedding property holds:
    \begin{equation*}
        \nicefrac{1}{2} \norm{\Ab \xb}_{p} \leq \norm{\Db_{p} \Sb_{p} \xb}_{p} \leq \nicefrac{3}{2} \norm{\Ab \xb}_{p}, \quad \forall \xb \in \R^{k}.  
    \end{equation*}
Using Remark 2.21~\cite{musco2021active_arxiv}, we get
\begin{equation*}
    s'_{p} = O \left( k ^{p/2} \ln^3{k} \right).
\end{equation*}
\end{lemma}
Notice that Remark 2.21~\cite{musco2021active_arxiv} removes the dependency of $s'_{p}$ on $\ln{d}$.
We also note that we can construct the matrices $\Db_{p}$ and $\Sb_{p}$ using Algorithm 1 of~\cite{musco2021active_arxiv}. Then, we solve the sampled $\ell_{p}$ regression problem $\tilde{\zb} = \arg \min_{\zb \in  \R^{k}}\norm{{\Db_{p} \Sb_{p} \Ab \zb - \Db_{p} \Sb_{p}  (\tb + \epsb_{\texttt{nq,p}})}}_{p}$, to get guarantees of the form: 
\begin{equation} \label{rla_reg_pre}
    \norm{\Ab \tilde{\zb} - (\tb + \epsb_{\texttt{nq,p}}) }_{p} \leq \gamma_p \norm{\Ab \hat{\zb} - (\tb + \epsb_{\texttt{nq,p}}) }_{p},
\end{equation}

for some constant $\gamma_p >1$ that depends on the choice of $3 \leq p < \infty$. In the above,
\begin{equation*}
    \hat{\zb} = \arg \min_{\zb \in \R^{k}} \norm{\Ab \zb -(\tb + \epsb_{\texttt{nq,p}})}_{p}.
\end{equation*}

\begin{lemma}[Theorem 3.2 in~\cite{musco2021active_arxiv}] \label{musco_2}
    Let $\Ab \in \R^{d \times k}$ and $(\tb + \epsb_{\texttt{nq,p}}) \in \R^{d}$. Let $3 \leq p < \infty$ be an integer and 
    $$\mathtt{OPT} = \min_{\zb \in \R^{k}}\norm{\Ab \zb - (\tb + \epsb_{\texttt{nq,p}})}_{p}.$$ 
    If $\tilde{\zb} = \arg \min_{\zb \in  \R^{k}}\norm{{\Db_{p} \Sb_{p} \Ab \zb - \Db_{p} \Sb_{p}  (\tb + \epsb_{\texttt{nq,p}})}}_{p}$, then, with probability at least $0.99$,
    \begin{equation*}
        \norm{\Ab \tilde{\zb} - (\tb + \epsb_{\texttt{nq,p}})}_{p} \leq 6 (200)^{1/p} \mathtt{OPT}.
    \end{equation*}
\end{lemma}
The quantity $(200)^{1/p}$ decreases very fast as $p$ increases. To put things into perspective for $p=3$,  $(200)^{1/p} < 6$. Importantly, to boost the success probability, we can compute
$O\left(\ln\left(\nicefrac{1}{\delta'}\right) \right)$ candidate solutions by running the algorithm implied by Lemma~\ref{musco_2} multiple times, for any $\delta' \in (0,1)$. Then, we can use Algorithm 3 of~\cite{musco2021active_arxiv} to select a solution that satisfies a slightly worse accuracy guarantee with high probability. More precisely:
\begin{lemma} [Theorem 3.3 in~\cite{musco2021active_arxiv}] \label{musco_3}
    Call Algorithm 3~\cite{musco2021active_arxiv} with inputs $O\left(\ln\left(\nicefrac{1}{\delta}\right) \right)$ candidate solutions computed by the algorithm of Lemma~\ref{musco_2}. Then, we get a vector $\tilde{\zb} \in \mathbb{R}^k$ such that, with probability at least $1-\delta'$,
    \begin{equation*}
         \norm{\Ab \tilde{\zb} - (\tb + \epsb_{\texttt{nq,p}})}_{p} \leq (18 (200)^{1/p} + 2) \mathtt{OPT}.
    \end{equation*}
\end{lemma}
Overall, the query complexity for $3\leq p <\infty$ is 
$$s_p = O\left( \ln\left( \nicefrac{1}{\delta'} \right) s'_{p} \right) = O\left( k^{p/2} \ln^{3}{k} \ln\left( \nicefrac{1}{\delta'} \right) \right).$$ 
Recall that we will need our bound to hold for all $n$ bidders via a union bound, so the failure probability must be reduced to $\delta' = \nicefrac{\delta}{n}$. Thus, our final sampling complexity $s_p$ is
\begin{equation*}
    s_{p} = O\left( k^{p/2} \ln^{3}{k}  \ln\left( \nicefrac{n}{\delta} \right) \right),
\end{equation*}
for any $\delta \in (0,1)$.

\subsection{The proof of Theorem~\ref{linear-algebra-thm}}\label{appendix:linear-algebra:theorem}
\begin{proof}[Proof of \cref{linear-algebra-thm}]
For notational simplicity, in this proof, we use $\tilde{\zb}$ instead of $\Qcal(\tb)$. Recall from Definition~\ref{dfn: model error} that we can write (dropping the index $i$ for simplicity): 
\begin{equation}\label{eqn:ppdapp1}
\tb  = \Ab \zb + \epsb_{\texttt{mdl,p}},
\end{equation}
where $\norm{\epsb_{\texttt{mdl,p}}}_{p} \leq \emdl$. We now use the active learning protocols of Sections~\ref{appendix:linear-algebra:query12} or~\ref{appendix:linear-algebra:querylargerthan2} to construct a solution $\tilde{\zb} \in \mathbb{R}^k$ such that
%
%
%
\begin{equation} \label{rla_reg}
    \norm{\Ab \tilde{\zb} - (\tb + \epsb_{\texttt{nq,p}}) }_{p} \leq \gamma_p \norm{\Ab \hat{\zb} - (\tb + \epsb_{\texttt{nq,p}}) }_{p},
\end{equation}
for some constant $\gamma_p >1$ that depends on the choice of $1 \leq p < \infty$. In the above,
\begin{equation*}
    \hat{\zb} = \arg \min_{\zb \in \R^{k}} \norm{\Ab \zb -(\tb + \epsb_{\texttt{nq,p}})}_{p}.
\end{equation*}
For $p=2$, eqn.~(\ref{eqn:pdappendix_pequal1f}) implies that $\gamma_2 = 6.5$. For $p=1$, Lemma~\ref{lem:pequalone_p_1} implies that $\gamma_1 = 1.5$. For $p\geq 3$, Lemma~\ref{musco_3} implies that $\gamma_p = 18 (200)^{1/p} + 2$. 

We now proceed to bound $\norm{\zb - \tilde{\zb}}_{p}$. To achieve this, we bound the left- and right-hand-sides of eqn.~\eqref{rla_reg}. Substitute $\tb = \Ab \zb + \epsb_{\texttt{mdl,p}}$ to the left hand side to get:
\begin{flalign} 
        \norm{\Ab \tilde{\zb} - (\tb + \epsb_{\texttt{nq,p}}) }_{p} &= \norm{\Ab(\tilde{\zb} - \zb) - (\epsb_{\texttt{nq,p}}+\epsb_{\texttt{mdl,p}})}_{p} \nonumber`\\
    & \geq \norm{\Ab(\zb - \tilde{\zb})}_{p} - \norm{\epsb_{\texttt{nq,p}}+\epsb_{\texttt{mdl,p}}}_{p}  \nonumber\\
    &\geq \norm{\Ab(\zb - \tilde{\zb})}_{p} - (\enq+\emdl).\label{lhs}
\end{flalign}
The first inequality follows from the reverse triangle inequality and the second from the assumptions on $\epsb_{\texttt{nq,p}}$ and $\epsb_{\texttt{mdl,p}}$. We now manipulate the right hand side of eqn.~\eqref{rla_reg}:
\begin{flalign}
        \gamma_p \norm{\Ab \hat{\zb} - (\tb + \epsb_{\texttt{nq,p}}) }_{p} &\leq  \gamma_p \norm{\Ab \zb - (\tb + \epsb_{\texttt{nq,p}}) }_{p} \nonumber\\
        &= \gamma_p \norm{-(\epsb_{\texttt{nq,p}}+\epsb_{\texttt{mdl,p}})}_{p} \nonumber \\
        & \leq \gamma_p (\enq+\emdl). \label{rhs}
\end{flalign}
The first inequality follows from the fact that $\hat{\zb}$ is the optimum solution of eqn.~\eqref{rla_reg} and the second equality follows from eqn.~(\ref{eqn:ppdapp1}). By combining eqns.~\eqref{rla_reg},~\eqref{lhs}, and~\eqref{rhs}, we get 
$$\norm{\Ab(\zb - \tilde{\zb})}_{p} \leq (\gamma_p+1)(\enq+\emdl).$$ 
Using eqn.~\eqref{gen_sigma}, we get $\norm{\Ab(\zb - \tilde{\zb})}_{p} \geq \sigma_{\min,p} \|\zb - \tilde{\zb}\|_p$. Let $c_p = \gamma_p+1$ to conclude:
\begin{equation} \label{p_bound}
    \norm{\zb - \tilde{\zb}}_{p} \leq \sigma_{\min,p}^{-1}(\Ab)\cdot c_p (\enq+\emdl). 
\end{equation}
Using the $\gamma_p$ values from above, we conclude that for $p=2$, $c_2 = 7.5$; for $p=1$, $c_1 = 2.5$; and for $p\geq 3$, $c_p = 18 (200)^{1/p} + 3$. (We did not optimize constants, but it is worth noting that reducing $c_p$ below two will need a different approach.) The failure probability of the theorem follows from the failure probability of eqn.~(\ref{rla_reg}), which needs to hold for all $n$ bidders with probability at least $1-\delta$. Applying a union bound and using the failure probabilities presented in Sections~\ref{appendix:linear-algebra:query12} and~\ref{appendix:linear-algebra:querylargerthan2} concludes the proof of the theorem.
\end{proof}

\subsection{Improving Theorem~2 of~\cite{cai2022recommender} and comparisons with our work}\label{appendix:cdlemma}

We revisit the proof of Theorem 2 of~\cite{cai2022recommender}, using our notation. We present a slightly improved analysis for the $\ell_\infty$ norm case. Let $\Db_{\infty} \Sb_{\infty}  \tb$ and $\Db_{\infty} \Sb_{\infty} (\tb + \epsb_{\texttt{nq},\infty})$ be the bidders' responses to the query protocol $\Qcal$ with and without the noise of the noisy query model.~\cite{cai2022recommender} makes the following assumptions, which are completely analogous to our work:
\begin{enumerate}
    \item $\norm{\epsb_{\texttt{nq},\infty}}_{\infty} \leq \enqin$ (query noise), and
    \item $\|\tb - \Ab \zb \|_{\infty} \leq \emdlin$ (model noise).
\end{enumerate}
%
~\cite{cai2022recommender} proceeds by solving the induced least squares (instead of $\ell_\infty$) regression problem as follows:
\begin{align}
     \tilde{\zb} &= \arg\min_{\zb \in \mathbb{R}^{k}} \| \Db_{\infty} \Sb_{\infty}  \Ab \zb - \Db_{\infty} \Sb_{\infty} (\tb+\epsb_{\texttt{nq},\infty} ) \|_{2} \nonumber\\
    &= \left((\Db_{\infty} \Sb_{\infty}  \Ab)^{T} \Db_{\infty} \Sb_{\infty}  \Ab \right)^{-1}(\Db_{\infty} \Sb_{\infty} \Ab)^{T} \Db_{\infty} \Sb_{\infty} (\tb+\epsb_{\texttt{nq},\infty} ) \label{lsq_2}\\
    &= ( \Db_{\infty} \Sb_{\infty} \Ab)^{\dagger}  \Db_{\infty} \Sb_{\infty} (\tb+\epsb_{\texttt{nq},\infty} ) \label{lsq_3}.
\end{align}
In the above we assume that the inverse of the $k \times k$ matrix $(\Db_{\infty} \Sb_{\infty}  \Ab)^{T} \Db_{\infty} \Sb_{\infty}  \Ab$ exists, which is also an implicit assumption in~\cite{cai2022recommender}. Thus, $\Db_{\infty} \Sb_{\infty}  \Ab \in \R^{s_{\infty} \times k}$ has full column rank equal to $k$. However, unlike~\cite{cai2022recommender}, we will use eqn.~\eqref{lsq_3} instead of eqn.~\eqref{lsq_2} in our analysis. This will result in improved bounds with respect to various quantities that arise in the analysis. We now proceed to bound $\|\tilde{\zb} - \zb \|_{\infty}$, as follows:
\begin{align*}
    \tilde{\zb} - \zb &= ( \Db_{\infty} \Sb_{\infty}  \Ab)^{\dagger} \left( \Db_{\infty} \Sb_{\infty}  (\tb+\epsb_{\texttt{nq},\infty} )- \Db_{\infty} \Sb_{\infty}  \Ab \zb\right)\\ 
    &= ( \Db_{\infty} \Sb_{\infty} \Ab)^{\dagger} \left(  \Db_{\infty} \Sb_{\infty} \epsb_{\texttt{nq},\infty} \right) + ( \Db_{\infty} \Sb_{\infty} \Ab)^{\dagger}  \Db_{\infty} \Sb_{\infty} (\tb - \Ab \zb). 
\end{align*}
The first equality follows by the assumption that $\Db_{\infty} \Sb_{\infty}  \Ab \in \R^{s_{\infty} \times k}$ has full column rank equal to $k$, which implies that
$(\Db_{\infty} \Sb_{\infty}  \Ab)^{\dagger}  \Db_{\infty} \Sb_{\infty} \Ab =\Ib.$
By applying the triangle inequality and using sub-multiplicativity properties of the $\ell_\infty$ norm, we get
\begin{align}
     \| \tilde{\zb} - \zb \|_{\infty} &\leq \enqin \norm{( \Db_{\infty} \Sb_{\infty} \Ab)^{\dagger}}_{\infty} \norm{\Db_{\infty}}_{\infty} \norm{\Sb_{\infty}}_{\infty}  \nonumber\\ 
     & \quad + \|( \Db_{\infty} \Sb_{\infty} \Ab)^{\dagger} \|_{\infty}  \norm{\Db_{\infty} }_{\infty} \|\Sb_{\infty} \|_{\infty}  \|\tb - \Ab \zb \|_{\infty} \nonumber\\ 
    &\leq \enqin \|(\Db_{\infty} \Sb_{\infty}  \Ab)^{\dagger} \|_{\infty} \norm{\Db_{\infty}}_{\infty} +  \emdlin \|(\Db_{\infty} \Sb_{\infty} \Ab)^{\dagger} \|_{\infty} \norm{\Db_{\infty}}_{\infty} \nonumber\\
    &= (\enqin+\emdlin)  \|( \Db_{\infty} \Sb_{\infty} \Ab)^{\dagger} \|_{\infty} \norm{\Db_{\infty}}_{\infty} \nonumber\\
    &\leq (\enqin+\emdlin) \sqrt{s_{\infty}} \norm{(\Db_{\infty} \Sb_{\infty}  \Ab)^{\dagger}}_{2} \norm{\Db_{\infty}}_{\infty} \nonumber\\
    &\leq \sqrt{s_{\infty}}(\enqin+\emdlin) \sigma_{\min}^{-1} \left( \Db_{\infty} \Sb_{\infty}  \Ab \right) \norm{\Db_{\infty}}_{\infty},\label{eqn:pdappendix222}
\end{align}
where for a matrix $\Xb \in \R^{m \times n}$ we have $\norm{\Xb}_{\infty} = \max_{i \in [m]}\sum_{j=1}^{n}{|\Xb_{ij} |}$ and $\norm{\Xb}_{2} = \max_{\norm{\yb}_{2} = 1}{\norm{\Xb \yb}_{2}} = \sigma_{\max}(\Xb)$. In the above we used that fact that $\norm{\Sb_{\infty}}_{\infty}=1$, since $\Sb_{\infty}$ is a sampling matrix and the property $\norm{\Xb}_{\infty} \leq \sqrt{n} \norm{\Xb}_{2}$ for any matrix $\Xb \in \R^{m \times n}$. The use of eqn.~\eqref{lsq_3} instead of eqn.~\eqref{lsq_2} improves the bound of~\cite{cai2022recommender} by avoiding an extra $\sqrt{k}$ term and a quadratic dependency on $\sigma_{\min} \left(\Db_{\infty} \Sb_{\infty} \Ab \right)$.

The bound of eqn.~(\ref{eqn:pdappendix222}) generalizes the bound of~\cite{cai2022recommender} to allow for sampling and rescaling of the revealed noisy bidder preferences, as well as the corresponding rows of the archetype matrix $\Ab$, following the lines of Algorithm~\ref{al1}. It is important to note that in order to get meaningful results using the bound of eqn.~(\ref{eqn:pdappendix222}), we need a lower bound on $\sigma_{\min} \left( \Db_{\infty} \Sb_{\infty}  \Ab \right)$ and an upper bound on $\norm{\Db_{\infty}}_{\infty}$.
The approach of~\cite{cai2022recommender} has no rescaling: it sets $\Db_\infty = \Ib$ and, therefore, $\norm{\Db_{\infty}}_{\infty} = 1$. However, this makes it hard to get non-trivial lower bound for the smallest singular value of the matrix $\Sb_{\infty} \Ab$. This matrix is simply a sample of rows of the matrix $\Ab$ without any rescaling and its smallest singular value could, in general, be arbitrarily close to zero. As a result, only special cases of the archetype matrix $\Ab$ were analyzed in~\cite{cai2022recommender}: for those special cases, the smallest singular value of a sample of rows of the input matrix (without rescaling) can be lower bounded. It is a well-known fact that if one were to form the sampling matrix $\Sb_\infty$ and the rescaling matrix $\Db_{\infty}$ using Algorithm~\ref{al1} and the row leverage scores of $\Ab$ as the sampling probabilities, the smallest singular value of $\Db_\infty\Sb_{\infty} \Ab$ would be close to the smallest singular value of $\Ab$. This, however, cannot be done without rescaling, which would necessitate an upper bound on $\norm{\Db_{\infty}}_{\infty}$. 
As a concrete example, if all the row leverage scores of $\Ab$ were equal, as would be the case if $\Ab$'s columns were a subset of the columns of a Hadamard matrix, the sampling probabilities would all be equal to $\nicefrac{1}{d}$. In this case, the diagonal rescaling matrix $\Db_{\infty}$ has entries that are propotional to the reciprocal of the sampling probabilities and $\norm{\Db}_\infty = \sqrt{\nicefrac{d}{s_{\infty}}}$. This makes the overall error bound depend on the number of bidder types $d$, since $s_{\infty}$ (the query complexity) is much smaller than $d$. This tension between lower-bounding the smallest singular value of $\Db_{\infty} \Sb_{\infty}  \Ab$ and upper bounding the infinity norm of $\Db_{\infty}$ does not seem easy to resolve by following the proposed analysis for arbitrary archetype matrices $\Ab \in \mathbb{R}^{d \times k}$.

\subsubsection{Discussing the assumptions of~\cite{cai2022recommender} for the archetype matrix}\label{app: discussion of A}

We conclude this section with a brief discussion of the three settings of Theorem~2 of~\cite{cai2022recommender} for the archetype matrix. First, for the deterministic structure, the assumption on the archetype matrix $\Ab \in \mathbb{R}^{d \times k}$ is that it has to contain a square $k\times k$ submatrix $\Cb$, which is both row- and column-diagonally-dominant. It is unclear whether this is a reasonable assumption in the context of archetype matrices; some connections with the non-negative matrix factorization literature are briefly discussed \cite{donoho2003does}. Importantly, it is unclear what the query protocol $\Qcal$ will be in this setting, since the query matrix $\Qb$ has to satisfy $\Qb \cdot \Ab = \Cb$ and it is not obvious how to recover $\Qb$ if $\Cb$ is unknown. The other two settings are both probabilistic and draw $k$ archetypes either from a $d$-dimensional Gaussian distribution, or as copies of a $d$-dimensional random vector. In both cases, the archetypes are drawn from distributions that are designed to construct matrices whose columns are approximately orthogonal, at least in expectation. This observation results in large values for $\sigma_{\min}(\Ab)$, and, under additional technical assumptions stated in Theorem 2 and Proposition~1 of~\cite{cai2022recommender}, large values for the smallest singular value of $\Sb_{\infty}  \Ab$. 
Our work directly characterizes the sample complexity of the protocol in terms of a single parameter of the archetype matrix that, intuitively, measures archetype independence with respect to $p$-norms. We provide explicit query protocols that leverage information in the archetype matrix and achieve the promised query complexity. We believe that this is a natural way to connect mechanism design with active learning for regression problems.


\end{document}